\documentclass[final,leqno,showlabe]{siamltex}
\usepackage{amsmath}
\usepackage{amssymb}
\usepackage{cases}
\usepackage{enumerate}
\usepackage{graphicx}
\usepackage{cleveref}
\usepackage[usenames]{color}
\usepackage{wrapfig}
\usepackage{stmaryrd}
\SetSymbolFont{stmry}{bold}{U}{stmry}{m}{n}
\newtheorem{thm}{Theorem}[section]

\newcommand{\norm}[1]{\lVert#1\rVert}

\def\subFS{\scriptscriptstyle{FS}}
\def\subVS{\scriptscriptstyle{VS}}

\renewcommand{\vec}[1]{\mathbf{#1}}

\title{A parametric finite element method for solid-state dewetting problems in three dimensions}
\author{Quan Zhao\thanks {Department of Mathematics, National University of Singapore, Singapore 119076 (quanzhao90@u.nus.edu). This author's research was supported by the Ministry of Education of Singapore grant R-146-000-247-114.}
\and Wei Jiang\thanks{Corresponding author. School of Mathematics and Statistics {\rm\&} Computational Science Hubei Key Laboratory, Wuhan University, Wuhan 430072, P.R. China (jiangwei1007@whu.edu.cn). This author's research was supported by the National Natural Science Foundation of China No. 11871384, and Natural Science Foundation of Hubei Province No. 2018CFB466.}
\and Weizhu Bao\thanks{Department of Mathematics, National University of Singapore, Singapore 119076 (matbaowz@nus.edu.sg, URL:
http://blog.nus.edu.sg/matbwz/). This author's research was supported by the Ministry of Education of Singapore grant R-146-000-290-114 and the National Natural Science Foundation of China No. 91630207.}
}

\date{}
\begin{document}

\maketitle
%%%%% Begin Abstract %%%%%%%%%%%

\begin{abstract}
We propose a parametric finite element method (PFEM) for efficiently solving the morphological evolution of solid-state dewetting of thin films on a flat rigid substrate in three dimensions (3D). The interface evolution of the dewetting problem in 3D is described by a sharp-interface model, which includes surface diffusion coupled with contact line migration.
A variational formulation of the sharp-interface model is presented, and a PFEM is proposed for spatial discretization. For temporal discretization, at each time step, we first update the position of the contact line according to the relaxed contact angle condition; then, by using the position of the new contact line as the boundary condition, we solve a linear algebra system resulted from the discretization of PFEM to obtain the new interface surface for the next step. The well-posedness of the solution of the PFEM is also established. Extensive numerical results are reported to demonstrate the accuracy and efficiency of the proposed PFEM and to show the complexities of the dewetting morphology evolution observed in solid-state dewetting experiments.
\end{abstract}
%%%%% end %%%%%%%%%%%

%%%%% Keywords %%%%%%%%%%%
%\pac{}
%\ams{35Q55, 65M70, 65N25, 65N35, 81Q05}

%\pac{68.35.-p, 68.55.Jk, 68.37.-d, 81.16.Rf}

\begin{keywords}
Solid-state dewetting, surface diffusion, moving contact line,
sharp-interface model, Cahn-Hoffman $\boldsymbol{\xi}$-vector.
\end{keywords}

\begin{AMS}
74H15, 74S05, 74M15, 65Z99
\end{AMS}

\pagestyle{myheadings} \markboth{Q. Zhao, W. Jiang and W. Bao}
{A parametric finite element method for solid-state dewetting in three dimensions}

%=============================================Introduction===============================================
%=====================================================================================================
\section{Introduction}

Solid-state dewetting is a ubiquitous phenomenon in materials science, and it describes the agglomeration of solid thin films into arrays of isolated particles on a substrate (e.g., see the review papers \cite{Thompson12, Leroy16}). In recent years, solid-state dewetting has found wide applications in thin film technologies, and it can be used to produce the controlled formation of an array of nanoscale particles, e.g., used in sensors \cite{Mizsei93} and as catalysts for carbon \cite{Randolph07} and semiconductor nanowire growth \cite{Schmidt09}. Recently, it has attracted extensive attention of many research groups, and has been widely studied from the experimental (e.g.,~\cite{Ye10b,Ye11a,Amram12,Rabkin14,Naffouti16,Kovalenko17}) and theoretical (e.g.,~\cite{Srolovitz86a,Wong00,Dornel06,Jiang12,Wang15,Jiang16,Bao17,Kim13,Zucker16}) points of view.

The dewetting of thin solid films deposited on substrates is similar to the dewetting phenomena of liquid films.
Although liquid-state wetting\slash dewetting problems have been extensively studied in the fluid mechanics(e.g.,~\cite{deGennes85,Qian06,Xu10,Xu11}), solid-state dewetting problems (i.e., surface diffusion-controlled geometric evolution) pose a considerable challenge for materials science, applied mathematics, and scientific computing. The major challenge comes from the difference of their mass transports. In general, surface diffusion has been recognized as the dominant mass transport for solid-state dewetting, and has played an essential role in determining the morphology evolution of solid thin films during the dewetting. The surface diffusion equation for the evolution of the film/vapor interface with isotropic surface energy (i.e., a constant, labeled as $\gamma_0$) was given by Mullins~\cite{Mullins57},
\begin{equation}
v_n=B\nabla_{_S}^2\mathcal{H},\quad\text{with}\quad B=\frac{D_s\gamma_0\Omega_0^2\nu}{k_B T},
\label{eqn:isoSF}
\end{equation}
where $v_n$ is the normal velocity of the film/vapor interface (surface), $D_s$ is the surface diffusivity, $k_B T$ is the thermal energy, $\nu$ is the number of diffusing atoms per unit area, $\Delta_s$ is the Laplace-Beltrami operator,
and $\mathcal{H}$ represents the mean curvature of the interface. For anisotropic surface energy (i.e., a function,
labeled as $\gamma=\gamma(\vec n)$ with $\vec n=(n_1,n_2,n_3)^T$ representing the unit outward normal orientation of the interface), it means that the surface energy (density) exhibits dependence on the crystalline orientation,  and \eqref{eqn:isoSF} can be readily
extended to the anisotropic case by replacing the mean curvature $\mathcal{H}$ with the weighted mean curvature $\mathcal{H}_\gamma$ as~\cite{Taylor92,Cahn94}
\begin{equation}
\mathcal{H}_\gamma = \nabla_{_S}\cdot\boldsymbol{\xi},
\label{eqn:anisSF}
\end{equation}
where $\nabla_{_S}$ is the surface gradient operator, and  $\boldsymbol{\xi}:=\boldsymbol{\xi}(\vec n)$ is well-known as the Cahn-Hoffman $\boldsymbol{\xi}$-vector~\cite{Cahn74,Hoffman72,Jiang18,Bao18a} which can be defined based on the homogeneous extension of $\gamma(\vec n)$ as
\begin{equation}
\boldsymbol{\xi}(\vec n) = \nabla\hat{\gamma}(\vec p)\Big|_{\vec p=\vec n},\qquad{\rm with}\;\hat{\gamma}(\vec p)=|\vec p|\gamma\left(\frac{\vec p}{|\vec p|}\right),\quad\forall\vec p\in\mathbb{R}^3\backslash\{\vec 0\},
\label{eqn:xivector}
\end{equation}
with $|\vec p|:=\sqrt{p_1^2+p_2^2+p_3^2}$, and $\vec p=(p_1,~p_2,~p_3)^T\in\mathbb{R}^3$.

Numerical simulations of geometric evolution equations have attracted considerate interest in decades, and different methods have been proposed in the literature for simulating the evolution of a closed curve\slash surface under mean curvature flow, surface diffusion flow, Willmore flow and etc. Theories on stable finite element methods for solving the flows of graphs~\cite{Bansch04, Deckelnick05, Xu09} have been fully studied. Unfortunately, these methods can not be directly applied to the case for general curves/surfaces (closed or open) due to the complicated governing geometric PDEs
with specific boundary conditions and unexpected deformation or topological events. Other front-tacking methods have been proposed to simulate evolutions for curves/surfaces, such as the marker-particle method~\cite{Du10, Leung11, Hon14}, and the parametric finite element method (PFEM)~\cite{Dziuk90, Bansch05, Pozzi08, Hausser05, Hausser07}. These methods are very efficient and render a very accurate representation of the interface compared to the phase field approach or level set approach. However, throughout the practical computation, these algorithms generally need complicated mesh regularizations or frequently re-meshing to improve the mesh quality for the discrete interface. To tackle this problem, Barrett {\it et al.} proposed a new novel parametric finite element method (e.g.,~\cite{Barrett07, Barrett07b, Barrett08, Barrett08JCP}), which has very good properties with respect to the distribution of mesh points. Precisely, their scheme introduced an implicit tangential motion for mesh points on the moving interface such that these mesh points automatically move tangentially along the interface and maintain good mesh properties, and this scheme has been extended for simulating the grain boundary motion and application of thermal grooving and sintering~\cite{Barrett10, Zhao17}.

Solid-state dewetting of thin films belongs to the evolution of an open curve\slash surface governed by surface diffusion and contact line migration~\cite{Jiang12,Wang15,Jiang16,Bao17,Bao18a}. In earlier years, the marker-particle method was firstly presented for solving sharp-interface models of solid-state dewetting in two dimensions (2D)~\cite{Wong00,Wang15} and three dimensions (3D)~\cite{Du10}. This method can be thought of as an explicit finite difference scheme, thus it imposes a very severe restriction on the time step for numerical stability. Furthermore, its extension to the 3D case is very tedious, inaccurate and time-consuming. For isotropic surface diffusion flow of a closed surface, B{\"a}nash {\it et al.} proposed a parametric finite element method together with a mesh regularization algorithm~\cite{Bansch05}; Barrett {\it et al.} then developed a simplified and novel variational formulation which leads to good mesh distribution properties and unconditional stability~\cite{Barrett08JCP, Barrett07}. These stable PFEMs were then generalized to the anisotropic case~\cite{Barrett07Ani, Barrett08Ani} for a special kind of anisotropy in terms of Riemannian metric form. Other related works for anisotropic flows in the literature can be found in~\cite{Hausser07, Pozzi08, Hildebrandt12} and references therein. Furthermore, the PFEMs have also been designed for simulating the evolution of solid-thin films on a substrate in 2D~\cite{Bao17b,Jiang18} and 3D case with axisymmetric geometry~\cite{Zhao19}. But how to design a PFEM for simulating solid-state dewetting problems in the full 3D is still urgent and challenging.

The goal of this paper is to extend our previous works~\cite{Bao17,Jiang18} from 2D to the 3D by using a variational formulation in terms of the Cahn-Hoffman $\boldsymbol{\xi}$-vector for simulating solid-state dewetting of thin films. More precisely, the main objectives are as follows: (i) to derive a variational formulation of the sharp-interface model for simulating solid-state dewetting problems in 3D~\cite{Bao18a}; (ii) to develop a PFEM for simulating the solid-state dewetting of thin films in 3D; (iii) to demonstrate the capability, efficiency and accuracy of the proposed PFEM; and (iv) to investigate many of the complexities which have been observed in experimental
dewetting of patterned islands on substrates, such as Rayleigh instability, pinch-off, edge retraction and corner mass accumulation.

The rest of the paper is organized as follows. In section 2, we briefly review a sharp-interface model for simulating solid-state dewetting problems in 3D, and then present a variational formulation of this sharp-interface model. In section 3, we discretize the variational formulation with a semi-implicit, mixed form PFEM. In section 4, extensive numerical results are reported to demonstrate the efficiency and accuracy of the PFEM scheme
and to show some interesting morphological evolution of solid-state dewetting in 3D. Finally, some conclusions are drawn in section 5.

\section{The model and its variational formulation}

In this section, we first review a sharp-interface model obtained recently by the authors~\cite{Bao18a} for simulating solid-state dewetting of thin films with isotropic/weakly anisotropic surface energies in 3D. Based on this model, we then propose a variational formulation via the Cahn-Hoffman $\boldsymbol{\xi}$-vector.

\subsection{The sharp-interface model}
\begin{figure}[!htp]
\centering
\includegraphics[width=0.8\textwidth]{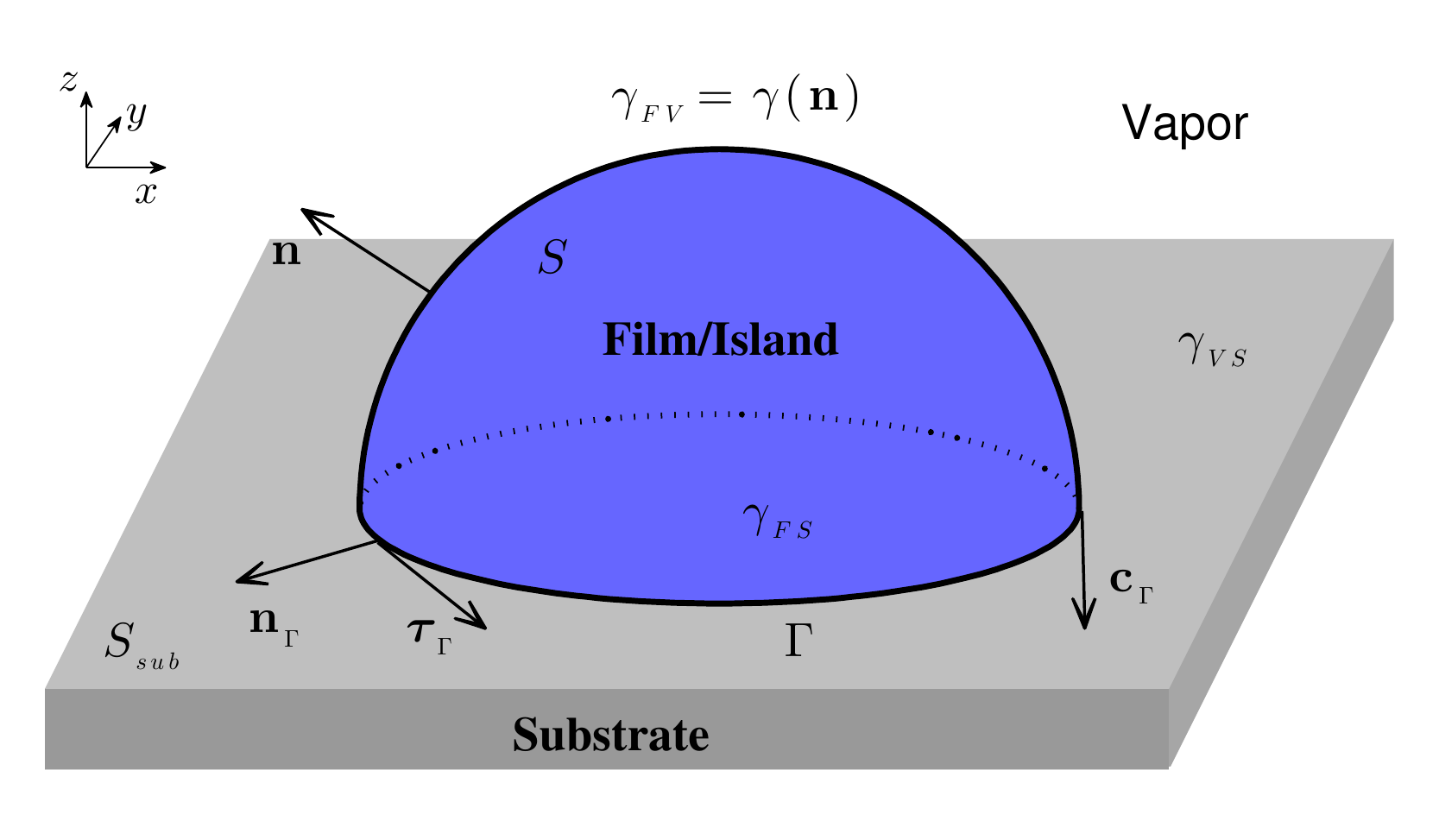}
\caption{A schematic illustration of solid-state dewetting of thin films on a flat substrate in 3D.}
\label{fig:model}
\end{figure}

As illustrated in Fig.~\ref{fig:model}, we consider that a solid thin film (shaded in blue) lies on a flat, rigid substrate (shaded in gray). The moving film/vapor interface, labeled as $S:=S(t)$, is represented by a time-dependent open surface with a plane curve boundary (i.e., the moving contact line, labeled as $\Gamma:=\Gamma(t)$) along the flat substrate $S_{\rm{sub}}$ (i.e., $Oxy$-plane). Let $U$ be a time-independent reference domain with $\vec u=(u_1,u_2)^T\in U \subseteq \mathbb{R}^2$, and assume that the moving surface $S(t):=\vec X(\vec u,~t)$ (with $\vec X=(x_1,~x_2,~x_3)^T$ or $(x,~y,~z)^T$) can be parameterized as
\begin{equation}
\vec X(\vec u,t)=(x(\vec u,t),~y(\vec u,t),~z(\vec u,t))^T:U\times[0,~T)\;\rightarrow\;\mathbb{R}^3.
\end{equation}
Furthermore, the moving contact line $\Gamma(t):=\vec X_{_\Gamma}(\cdot,~t)$ can be also parameterized over $\partial U$ as
\begin{equation}
\vec X_{_\Gamma}(\cdot,t)=(x_{_\Gamma}(\cdot,t),~y_{_\Gamma}(\cdot,t),~z_{_\Gamma}(\cdot,t))^T:\partial U\times[0,~T)\rightarrow\mathbb{R}^3.
\end{equation}
By using the approach in~\cite{Bao18a}, we can obtain a sharp-interface model for simulating solid-state dewetting of thin films with isotropic/weakly anisotropic surface energies in 3D as the following dimensionless form:
%\begin{subequations}
%\begin{numcases}{}
%\label{eqn:model1}
%\partial_t\vec X = \Delta_{_S}\mu\,\vec n,\qquad t>0,\\
%\mu = \nabla_{_S}\cdot\boldsymbol{\xi},\qquad\boldsymbol{\xi}(\vec n) = \nabla\hat{\gamma}(\vec p)\Big|_{\vec p=\vec n},
%\label{eqn:model2}
%\end{numcases}
%\end{subequations}
\begin{eqnarray}
\label{eqn:model1}
&&\partial_t\vec X=\Delta_{_S}\mu\;\vec n,\qquad t>0,\\
&&\mu=\nabla_{_S}\cdot\boldsymbol{\xi},\qquad \boldsymbol\xi(\vec n)=\nabla\hat{\gamma}(\vec p)\Big|_{\vec p=\vec n},
\label{eqn:model2}
\end{eqnarray}
where $\mu$ is the chemical potential, $\vec n =(n_1,~n_2,~n_3)^T$ is the unit outer normal vector of the moving surface $S$,  $\boldsymbol{\xi}(\vec n)=(\xi_1,~\xi_2,~\xi_3)^T$ represents the Cahn-Hoffman $\boldsymbol{\xi}$-vector associated with the surface energy density $\gamma(\vec n)$ (see Eq.~\eqref{eqn:xivector}), and $\Delta_{_S}:=\nabla_{_S} \cdot \nabla_{_S}$ is the Laplace-Beltrami operator defined on $S$. The initial condition is given as $S_0$ with boundary $\Gamma_0$ such that
\begin{equation}
S_0:=\vec X(\vec u,0)=\vec X_0(\vec u)=(x_0(\vec u),~y_0(\vec u),~z_0(\vec u))^T.
\label{eqn:initialcondition}
\end{equation}
The above governing equations are subject to the following boundary conditions:

(i) contact line condition
\begin{equation}
\label{eqn:boundcon1}
z_{_\Gamma}(\cdot,t) = 0,\qquad t\ge 0;
\end{equation}

(ii) relaxed contact angle condition
\begin{equation}
\partial_t \vec X_{_\Gamma}=-\eta\bigl(\vec c_{_\Gamma}^\gamma\cdot\vec n_{_\Gamma} - \sigma\bigr)\vec n_{_\Gamma},\qquad t\geq0;
\label{eqn:boundcon2}
\end{equation}

(iii) zero-mass flux condition
\begin{equation}
(\vec c_{_\Gamma}\cdot\nabla_{_S}\,\mu)\Big|_\Gamma=0,\qquad t\geq0.
\label{eqn:boundcon3}
\end{equation}
Here, $0<\eta<\infty$ represents the contact line mobility, and the vector $\vec c_{_\Gamma}^\gamma$ is defined as a linear combination of $\vec c_{_\Gamma}$ and $\vec n$,
\begin{equation}
 \vec c_{_\Gamma}^\gamma:=(\boldsymbol{\xi}\cdot\vec n)\,\vec c_{_\Gamma}-(\boldsymbol{\xi}\cdot\vec c_{_\Gamma})\,\vec n,
 \label{eqn:cgamma}
 \end{equation}
where $\vec c_{_\Gamma}=(c_{_\Gamma,_1},~c_{_\Gamma,_2},~c_{_\Gamma,_3})^T$ is called as the unit co-normal vector, which is normal to $\Gamma$, tangent to the surface $S$ and points outwards,  $\vec n_{_\Gamma}=(n_{_\Gamma,_1},~n_{_\Gamma,_2},~0)^T$ is the unit outer normal vector of $\Gamma$ on the substrate (as shown in Fig.~\ref{fig:model}), and $\sigma:=({\gamma_{_{\subVS}}-\gamma_{_{\subFS}}})/{\gamma_0}$ is a (dimensionless) material constant with $\gamma_0$ being the characteristic unit for surface energy, where the two constants $\gamma_{_{\subVS}}$ and $\gamma_{_{\subFS}}$ represent the vapor/substrate and film/substrate surface energy densities, respectively.

Condition (i) (i.e., Eq.~\eqref{eqn:boundcon1}) ensures that the contact line moves along the substrate during time evolution. Condition (ii) prescribes a contact angle condition along the moving contact line.  In order to understand this condition, we may consider two limiting cases as $\eta=0$ and $\eta=\infty$: (i) when $\eta=0$, the contact line moving velocity is zero, and we prescribe a fixed boundary condition such that the contact line does not move; and (ii) when $\eta\rightarrow\infty$, as we always assume that the moving velocity should be finite, condition (ii) will reduce to the so-called anisotropic Young equation~\cite{Bao18a,Bao17b}
\begin{equation}
\vec c_{_\Gamma}^\gamma\cdot\vec n_{_\Gamma} - \sigma = 0.
\label{eqn:staticangle}
\end{equation}
which prescribes an equilibrium contact angle condition. Therefore, condition (ii) actually allows a relaxation process for the dynamic contact angle evolving to its equilibrium contact angle~\cite{Wang15,Jiang16}. The last condition (iii) ensures that the total volume/mass of the thin film is conserved during the evolution, i.e., no-mass flux at the moving contact line. We remark that if the moving surface has more than one closed curve as its boundary (see examples in Fig.~\ref{fig:hollow12}), then the boundary conditions~\eqref{eqn:boundcon1}-\eqref{eqn:boundcon3} should be satisfied on each boundary curve.

The above sharp-interface model~\eqref{eqn:model1}-\eqref{eqn:model2} with boundary conditions~\eqref{eqn:boundcon1}-\eqref{eqn:boundcon2} are derived based on the consideration of thermodynamic variation~\cite{Bao17b,Bao18a}, and therefore, it naturally satisfies the thermodynamic-consistent physical law. More precisely, the total (dimensionless) free energy of the system, including the interface energy $W_{\rm int}$ and substrate energy $W_{\rm sub}$, can be written as~\cite{Bao17b,Bao18a}
\begin{equation}
W(t):=W_{\rm int} + W_{\rm sub}=\int_{S(t)}\gamma(\vec n)\,d\,S - \sigma A(\Gamma),
\label{eqn:energy}
\end{equation}
where $A(\Gamma)$ denotes the surface area enclosed by the contact line curve $\Gamma$ on the substrate. It can be easily shown that during the evolution which is governed by the above sharp-interface model~\cite{Bao18a}, the total volume of the thin film is conserved and the total free energy satisfies the following dissipation law
\begin{equation}
\frac{d}{dt}W(t) = -\int_{S(t)}|\nabla_{_S}\mu|^2\;dS -\eta\int_{\Gamma(t)}\Bigl(\vec c_{_\Gamma}^\gamma\cdot\vec n_{_\Gamma} - \sigma\Bigr)^2\;d\Gamma \leq 0, \qquad t\ge0.
\end{equation}

\subsection{The variational formulation}

Let $S:=S(t)\in C^2(U)$ be a smooth surface with smooth boundary $\Gamma:=\Gamma(t)$ , and assume that $f\in C(\bar{S})$. Denote the surface gradient operator as $\nabla_{_S}:=(\underline{D}_1,~\underline{D}_2,~\underline{D}_3)^T$, then the integration by parts on an open smooth surface $S$ with smooth boundary $\Gamma$ can be written as~\cite{Bao18a,Dziuk13}
\begin{equation}
\label{eqn:integrationbyparts}
\int_S\underline{D}_i f\;dS = \int_S f\mathcal{H}\; n_i\;dS + \int_\Gamma f c_{_\Gamma,_i}\;d\Gamma,
\end{equation}
where $\mathcal{H}=\nabla_{_S}\cdot\vec n$ is the mean curvature of the surface $S$ and $\vec c_{_\Gamma}=(c_{_\Gamma,_1},~c_{_\Gamma,_2},~c_{_\Gamma,_3})^T$ is the co-normal vector defined above. Following the above formula, we can naturally define the derivative $\nabla_{_S} f$ in the weak sense. Then, we can define the functional space $L^2(S)$ as
\begin{equation}
L^2(S):= \Bigl\{f: S\rightarrow\mathbb{R},\; {\rm and}\;\norm{f}_{L^2(S)}=\Bigl(\int_S f^2\,dS\,\Bigr)^{\frac{1}{2}}<+\infty\Bigl\},
\end{equation}
equipped with the $L^2$ inner product for any scalar or vector-valued functions $f_1,\,f_2$ defined over the surface $S$ as follows
\begin{equation}
\big<f_1,~f_2\big>_S := \int_S f_1\cdot f_2 \;dS.
\end{equation}
The Sobolev space $H^1(S)$ can be naturally defined as
\begin{equation}
H^1(S):= \Bigl\{f: S\rightarrow\mathbb{R},\;f\in L^2(S),\; \underline{D}_if\in L^2(S),\quad\forall 1\leq i\leq 3\Bigr\},
\end{equation}
equipped with the norm $\norm{f}_{H^1(S)}:=\Bigl(\norm{f}_{L^2(S)}^2 + \norm{\nabla_{_S}f}_{L^2(S)}^2\Bigr)^{\frac{1}{2}}$. Furthermore, if we denote $T_{_S}: H^1(S)\rightarrow L^2(\Gamma)$ as the trace operator, we can define the following functional space with the homogeneous Dirichlet boundary condition:
\begin{equation}
H_0^1(S):=\Bigl\{f: f\in H^1(S),\quad T_{_S}f = 0\Bigr\}.
\end{equation}

Therefore, we can define the following functional space which will be used for the solution of the solid-state dewetting problem as
\begin{equation}
H_{\alpha}^1(U):=\Bigl\{\varphi\in H^1(U),\;\varphi\Big|_{\partial U}= g\Bigr\},
\end{equation}
where the function $g\in L^2(\partial U)$ is given.
From these definitions, it should be noted that $H_0^1(U)$ denotes the functions in $H^1(U)$ with trace being zeros.

We now propose the following variational formulation for the sharp-interface model~\eqref{eqn:model1}-\eqref{eqn:model2} with the boundary conditions~\eqref{eqn:boundcon1}-\eqref{eqn:boundcon3} as: given the initial surface $S_0:=\vec X_0$ with its boundary $\Gamma_0$ defined in~\eqref{eqn:initialcondition}, find its evolution surfaces $S(t):=\vec X(\cdot,t)\in H_\alpha^1(U)\times H_\beta^1(U)\times H_0^1(U)$, and the chemical potential $\mu(\cdot,t)\in H^1(S)$ such that
\begin{subequations}
\begin{align}
\label{eqn:aniweakform1}
&\big<\partial_t\vec X\cdot\vec n,~\varphi\big>_S+ \big<\nabla_{_S}\mu,~\nabla_{_S}\varphi\big>_S = 0,\qquad\forall\varphi\in H^1(S),\\[1em]
&\big<\mu,~\vec n\cdot\boldsymbol{\omega}\big>_S - \sum_{k=1}^3\big<\gamma(\vec n)\nabla_{_S}x_k,~\nabla_{_S}\omega_k\big>_S \nonumber\\
&\qquad\qquad\qquad+\;\sum_{k,l=1}^3\big<\xi_k\nabla_{_S}x_k,~n_l\nabla_{_S}\omega_l\big>_S=0,\qquad\forall\boldsymbol{\omega}\in\vec [H^1_0(S)]^3,
\label{eqn:aniweakform2}
\end{align}
\end{subequations}
where $\alpha,~\beta$ represents the $x,~y$-coordinates of the moving contact line at time $t$, i.e., $\alpha=x_{_\Gamma}(\cdot,t),\;\beta=y_{_\Gamma}(\cdot,t)$, and $\boldsymbol{\xi}(\vec n)=(\xi_1,~\xi_2,~\xi_3)^T$ represents the Cahn-Hoffman $\boldsymbol{\xi}$-vector associated with the surface energy density $\gamma(\vec n)$
(see the definition in~\eqref{eqn:xivector}). Here, $\Gamma(t)=\vec X_{_\Gamma}(t)=(x_{_\Gamma}(\cdot,t),~y_{_\Gamma}(\cdot,t),~0)^T$ is jointly determined by the relaxed angle boundary condition~\eqref{eqn:boundcon2} in the above weak formulation.

In the above weak formulation, \eqref{eqn:aniweakform1} can be obtained by reformulating~\eqref{eqn:model1} as $\partial_t\vec X\cdot\vec n =\Delta_{_S}\mu$, multiplying a scalar test function $\varphi\in H^1(S)$, integrating over $S(t)$, integration by parts and noting the zero-mass flux boundary condition~\eqref{eqn:boundcon3}; Similarly,  by multiplying $n_l$ to Eq.~\eqref{eqn:model2}, we obtain the following equation~\cite{Deckelnick05}
\begin{equation}
\mu\,n_{l} = (\nabla_{_S}\cdot\boldsymbol{\xi})\,n_{l} =  \underline{D}_k(\xi_k\,n_l) - \underline{D}_k(\gamma(\vec n)\underline{D}_k\,x_l) -\gamma(\vec n)\mathcal{H}\,n_l,\quad l=1,2,3,
\label{eqn:vamedit1}
\end{equation}
where summation over $k$ is from $1$ to $3$.
By multiplying~\eqref{eqn:vamedit1} with $\omega_l$ on both sides, summation over $l=1,2,3$,
integrating over $S$ and integration by parts, we can obtain ~\eqref{eqn:aniweakform2}. For more details, please refer to~\cite{Deckelnick05}. We note that~\eqref{eqn:aniweakform2} has also been used in some works related to anisotropic geometric evolution equations~\cite{Barrett08Ani,Pozzi08,Hildebrandt12}.

In the isotropic case, i.e., $\gamma(\vec n)\equiv 1$, we have $\boldsymbol{\xi}(\vec n)=\vec n$. By using the fact that $\underline{D}_kx_l = \delta_{kl} - n_k\,n_l$,  we can obtain $\sum_{k,l=1}^3\big<n_k\nabla_{_S}x_k,~n_l\nabla_{_S}\omega_l\big>_S=0$.
%the following for the third term in Eq.~\eqref{eqn:aniweakform2} as
%\begin{eqnarray}
%&&\sum_{l,k=1}^3n_k\,n_l(\nabla_{_S}x_k\cdot\nabla_{_S}\omega_l) = \sum_{l,k,g=1}^3n_k\,n_l\underline{D}_gx_k\underline{D}_g\omega_l\nonumber\\
%&&=\sum_{l,k,g=1}^3n_k\,n_l(\delta_{gk} - n_g\,n_k)\underline{D}_g\omega_l\nonumber\\
%&&=\sum_{l,k=1}^3n_k\,n_l\underline{D}_k\omega_l - \sum_{l,k,g=1}^3n_k\,n_l\,n_g\,n_k\underline{D}_g\omega_l=0.
%\end{eqnarray}
Therefore, Eq.~\eqref{eqn:aniweakform2} will reduce to the variational formulation of the curvature term related to the Laplace-Beltrami operator~\cite{Bansch05,Barrett08JCP}
\begin{equation}
\big<\mu,~\vec n\cdot\boldsymbol{\omega}\big>_S - \sum_{k=1}^3\big<\nabla_{_S}x_k,~\nabla_{_S}\omega_k\big>_S=0,\quad\forall\boldsymbol{\omega}\in\vec (H^1_0(S))^3.
\label{eqn:isoweakform}
\end{equation}
In general, it is not easy to obtain the energy stability based on the discretization of the variational formulation defined in~\eqref{eqn:aniweakform1}-\eqref{eqn:aniweakform2}.
Specifically, in the isotropic case, the stability bound for the discretization of \eqref{eqn:isoweakform} has been established for the evolution of a closed surface~\cite{Bansch05, Barrett08JCP}. Based on our numerical experiments, the variational formulation defined in~\eqref{eqn:aniweakform1}-\eqref{eqn:aniweakform2} and its PFEM
perform very well in terms of stability, efficiency and accuracy in practical computations.

\section{The parametric finite element approximation}
In this section, based on the variational formulation \eqref{eqn:aniweakform1}-\eqref{eqn:aniweakform2},  we discretize the problem via a semi-implicit parametric finite element method, and prove the well-posedness of the discrete scheme.

To present the PFEM for the variational problem, we first take the time step as
$0=t_0<t_1<t_2<\cdots<t_{M}$,
and denote time steps as $\tau_m=t_{m+1}-t_{m}$ for $0\leq m\leq M-1$. In the spatial level, we assume that the evolution surfaces $\{S(t_m)\}_{m=1}^M$ are discretized by polygonal surfaces $\{S^m\}_{m=1}^M$ such that
\begin{equation}
S^m=\bigcup_{j=1}^N \bar{D}_j^m,\quad{\rm where}\quad\{D_j^m\}_{j=1}^N\quad{\text {are mutually disjoint triangles}}.
\end{equation}
Here, we assume that the discrete surface $S^m$ has $K$ different vertices (labeled as $\{\vec q_k^m\}_{k=1}^K$), and the boundary of $S^m$ is a closed polygonal curve $\Gamma^m = \bigcup_{j=1}^{N_c}\bar{h}_j^m$,
where $\{h_j^m\}_{j=1}^{N_c}$ are a sequence of connected line segments which is positively oriented,
i.e., if you walk along the direction of the oriented boundary, the surface is at your left side.
Moreover, we have the following assumption about the polygonal surface at each time step such that
\begin{equation}
|D_j^m|>0,\quad 1\leq j\leq N,\quad 0\leq m\leq M-1,
\label{eqn:assumption}
\end{equation}
which ensures that vertices of polygonal surface will not merge during the evolution.

We can define the following finite dimensional spaces over $\Gamma^m$ and $S^m$ as
\begin{subequations}
\begin{align}
& V^h(\Gamma^m):=\Bigl\{\varphi\in C(\Gamma^m,~\mathbb{R}):\;\varphi\Big|_{h_j^m}\in \boldsymbol{P}_1,\;\forall 1\leq j\leq N_c\Bigr\}\subset H^1(\Gamma^m),\\
& V^h(S^m):=\Bigl\{\varphi\in C(S^m,~\mathbb{R}):\;\varphi\Big|_{D_j^m}\in\boldsymbol{P}_1,\;\forall 1\leq j\leq N\Bigr\}\subset H^1(S^m),
\end{align}
\end{subequations}
where $\boldsymbol{P}_1$ denotes all polynomials with degrees at most $1$, which yields piecewise linear functions on each element. If $g\in V^h(\Gamma^m)$, we can define the finite element space on $S^m$ with boundary value given by a function $g$ as
\begin{equation}
\mathcal{V}_g^h(S^m):=\Bigl\{\varphi\in V^h(S^m):\;\varphi\Big|_{\Gamma^m}=g\Bigr\}.
\end{equation}
Again, for simplicity of notations, we denote $\mathcal{V}_0^h$ as space $V^h(S^m)$ with zero values on the boundary $\Gamma^m$.

Now, we can define the following mass-lumped inner product to approximate the integration on $S^m$ as
\begin{eqnarray}
\label{eqn:masslumpnorm}
 \big<f_1,~f_2\big>_{m}^h = \frac{1}{3}\sum_{j=1}^N |D_j^m|\sum_{k=1}^3 f_1\left((\vec q_{j_{_k}}^m)^-\right)\cdot f_2\left((\vec q_{j_{_k}}^m)^-\right),
 \end{eqnarray}
where $|D_j^m|$ is the area of the triangle $D_j^m$, and $f_1,~f_2$ are two scalar or vector functions defined on $S^m$ with possible jumps across each edge of the triangle in 3D. We define the one-sided limit $f_1((\vec q_{j_{_k}}^m)^-)$ as the limit of $f_1(\vec x)$ when $\vec x$ approaches towards $\vec q_{j_{_k}}^m$ from the triangle surface $D_j^m$, i.e., $f_1((\vec q_{j_{_k}}^m)^-)=\lim\limits_{D_j^m\ni\vec x\rightarrow\vec q_{j_{_k}}^m }f_1(\vec x)$.

We assume that $\{\vec q_{j_1}^m,~\vec q_{j_2}^m,~\vec q_{j_3}^m\}$ are the three vertices of the triangle surface $D_j^m$ and ordered in the anti-clockwise direction when viewing from top to bottom. It should be noted that the normal vector $\vec n^m=(n^{m}_1,~n^{m}_2,~n^{m}_3)^T$ of the surface $S^m$ is a step function with discontinuities across the edges of each triangle surface. Let ${\vec n_j^m}$ be the unit normal vector on $D_j^m$, and we can numerically evaluate it as
\begin{equation}
\vec n_j^m:=\vec n^m\Big|_{D_j^m}= \frac{(\vec q_{j_2}^m - \vec q_{j_1}^m)\times (\vec q_{j_3}^m-\vec q_{j_1}^m)}{|(\vec q_{j_2}^m - \vec q_{j_1}^m)\times (\vec q_{j_3}^m-\vec q_{j_1}^m)|},\quad\forall 1\leq j\leq N.
\label{eqn:NumericalNormal}
\end{equation}
\begin{figure}[!htp]
\centering
\includegraphics[width=0.90\textwidth]{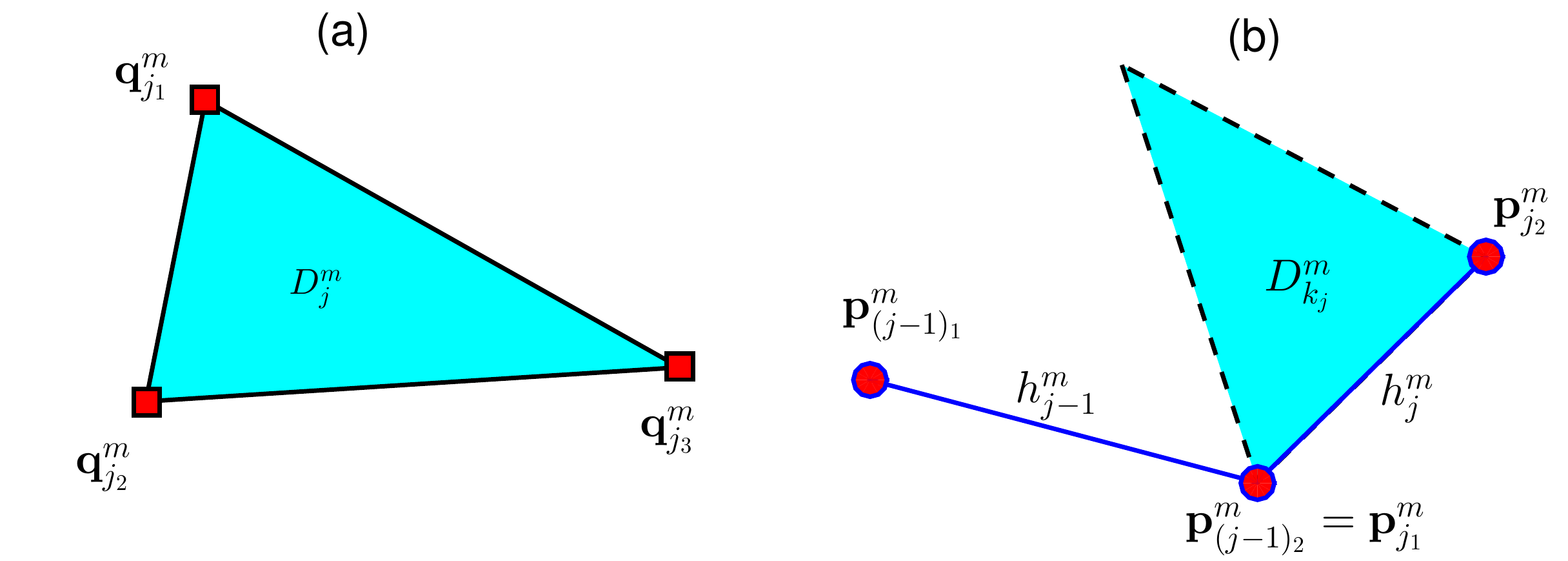}
\caption{A schematic illustration of surface triangle mesh when viewing from top to down: (a) a triangle mesh with no edges on the boundary; (b) a triangle mesh with an edge (shown in blue) on the boundary.}
\label{fig:Mesh}
\end{figure}
For the discrete boundary curve $\Gamma^m$, it is a closed plane curve and consists of a sequence of connected line segments on the substrate ($Oxy$-plane). We assume that $\{\vec p_{j_1}^m,~\vec p_{j_2}^m\}$ are the two vertices of a line segment $h_j^m$ which are ordered according to the orientation of the curve. Let $\vec n_{_{\Gamma}}^m$ denote the unit normal vector of the boundary curve $\Gamma^m$ along the substrate, then $\vec n_{_{\Gamma}}^m$ is also a step function with discontinuities across the vertices of each line segment. Let $\vec n_{_{\Gamma,j}}^m$ represent the unit normal vector of $\Gamma^m$ on the line segment $h_j^m$, then
\begin{equation}
\vec n_{_{\Gamma,j}}^m =\vec n_{_{\Gamma}}^m\Big|_{h_j^m}= \frac{(\vec p_{j_2}^m-\vec p_{j_1}^m)\times \vec e_3}{|(\vec p_{j_2}^m-\vec p_{j_1}^m)\times \vec e_3|},\quad\forall 1\leq j\leq N_c,
\label{eqn:NumericalGammaNormal}
\end{equation}
where the unit vector $\vec e_3=(0,~0,~1)^T$. Similarly, $\vec c_{_{\Gamma}}^m$ is the unit co-normal vector defined on the polygonal curve $\Gamma^m$ along the substrate, and it is also a step function which can be numerically evaluated as
\begin{equation}
\vec c_{_{\Gamma,j}}^m=\vec c_{_{\Gamma}}^m\Big|_{h_j^m} = \frac{(\vec p_{j_2}^m-\vec p_{j_1}^m)\times \vec n_{k_j}^m}{|(\vec p_{j_2}^m-\vec p_{j_1}^m)\times\vec n_{k_j}^m|},\quad\forall 1\leq j\leq N_c,
\label{eqn:NumericalCoNormal}
\end{equation}
where $\vec n_{k_j}^m$ is the unit outer normal vector of the triangle surface $D_{k_j}^m$ which contains the line segment $h_j^m$ (as shown in Fig.~\ref{fig:Mesh}(b)).

Let $S^m:=\vec X^m$ and $\Gamma^m:=\vec X_{_{\Gamma}}^m=(x_{_\Gamma}^{m},~y_{_\Gamma}^{m},~0)^T$ be the numerical approximations of the moving surface $S(t_m):=\vec X(\cdot,~t_m)$ and its boundary line $\Gamma(t_m):=\vec X_{_\Gamma} (\cdot,~t_m)$, respectively. Take $S^0=\vec X^0 \in\mathcal{V}^h_{\alpha_0}(S^m)\times\mathcal{V}^h_{\beta_0}(S^m)\times\mathcal{V}^h_0(S^m)$ with $\alpha_0,~\beta_0\in V^h(\Gamma^0)$ as the numerical approximations of $x_{_\Gamma}(\cdot,0),~y_{_\Gamma}(\cdot,0)$, respectively.

Then, a semi-implicit parametric finite element method for the
variational problem \eqref{eqn:aniweakform1}-\eqref{eqn:aniweakform2} can be stated as: given $S^0=\bigcup_{j=1}^N\bar{D}_j^0$ which is an initial polygonal surface and its boundary curve $\Gamma^0=\bigcup_{j=1}^{N_c}\bar{h}_j^0$, for $m\geq 0$, find a sequence of polygonal surfaces $S^{m+1}:=\vec X^{m+1}\in\mathcal{V}^h_{\alpha}(S^m)\times\mathcal{V}^h_{\beta}(S^m)\times\mathcal{V}_0^h(S^m)$, and chemical potentials $\mu^{m+1}\in V^h(S^m)$ such that
\begin{subequations}
\begin{align}
\label{eqn:fullanisoform1}
&\Big<\frac{\vec X^{m+1}-\vec X^m}{\tau_m},~\varphi_h\vec n^m\Big>_{m}^h + \big<\nabla_{_{S}}\mu^{m+1},~\nabla_{_{S}}\varphi_h\big>_{m}^h = 0,\qquad\forall \varphi_h\in V^h(S^m),\\
&\big<\mu^{m+1},~\vec n^m\cdot\boldsymbol{\omega}_h\big>_{m}^h - \sum_{l=1}^3\big<\gamma^m\nabla_{_{S}} x_l^{m+1},\nabla_{_{S}}\omega_{h,l}\big>_{m}^h =\mathcal{G}^m,\;\forall \boldsymbol{\omega}_h\in [\mathcal{V}_0^h(S^m)]^3,
\label{eqn:fullanisoform2}
\end{align}
\end{subequations}
where $\gamma^m$ and $\mathcal{G}^m$ are explicitly calculated as
\begin{equation}
\gamma^m=\gamma(\vec n^m),\qquad\mathcal{G}^m=-\sum_{k,l=1}^3\big<\xi^{m}_k\nabla_{_S}x^{m}_k,~n^{m}_l\,\nabla_{_S}\omega_{h,l}\big>_{m}^h,
\end{equation}
with $\boldsymbol{\xi}^m=\boldsymbol{\xi}(\vec n^m)=(\xi^{m}_1,~\xi^{m}_2,~\xi^{m}_3)^T$, $\boldsymbol{\omega_h}=(\omega_{h,1},~\omega_{h,2},~\omega_{h,3})^T$, and $\alpha,~\beta$ are the $x,~y$-coordinates of the contact line $\Gamma^{m+1}$, i.e., $\alpha:=x_{_\Gamma}^{{m+1}},~\beta :=y_{_\Gamma}^{{m+1}}$.

We note here that the boundary curve $\Gamma^{m+1}$ is first updated from $\Gamma^m$ by explicitly solving
the relaxed contact angle condition defined in Eq.~\eqref{eqn:boundcon2}, and then by using $\Gamma^{m+1}$ as the Dirichlet boundary condition, we solve the above PFEM to obtain the new polygonal surface $S^{m+1}$. More precisely, the algorithm for updating $\Gamma^{m+1}$ can be described as (shown in Fig.~\ref{fig:updateGamma}):
\begin{itemize}
\item Calculate $\vec n^m_{k_j}$, $\vec n_{_{\Gamma,j}}^m$ and $\vec c_{_{\Gamma,j}}^{m}$ via \eqref{eqn:NumericalNormal}, \eqref{eqn:NumericalGammaNormal} and \eqref{eqn:NumericalCoNormal}, and then by
    using forward Euler scheme to approximate the relaxed contact angle condition, we can obtain
$\lambda_j^m$ and $\vec V_j^m$ for each line segment $h_j^m$ as
\begin{equation*}
\lambda_j^m := -\tau_m\,\eta\, (\vec c_{_{\Gamma,j}}^{\gamma,m}\cdot\vec n_{_{\Gamma,j}}^m-\sigma),\qquad \vec V_j^m :=\lambda_j^m\,\vec n_{_\Gamma,_j}^m,\quad 1\leq j\leq N_c,
\label{eqn:Gammavelocity}
\end{equation*}
where $\vec c_{_{\Gamma,j}}^{\gamma,m}:=(\boldsymbol{\xi}(\vec n^{m}_{k_j})\cdot\vec n^m_{k_j})\;\vec c_{_{\Gamma,j}}^m -(\boldsymbol{\xi}(\vec n^m_{k_j})\cdot\vec c_{_{\Gamma,j}}^m)\;\vec n^m_{k_j}$;

\item If $\vec n_{_\Gamma,_{j-1}}^m\parallel\vec n_{_\Gamma,_j}^m$, we update the segmentation point $\vec p_{j_1}^m$ by moving along the displacement vector $\frac{1}{2}(\vec V_{j-1}^m + \vec V_j^m)$;

\item If $\vec n_{_\Gamma,_{j-1}}^m\nparallel\vec n_{_\Gamma,_j}^m$, we first move each line segment $h^m_j$ along its
normal direction by an increment vector $\vec V_j^m=\lambda_j^m\,\vec n_{_\Gamma,_j}^m$, then calculate the intersection point of the updated adjacent edges, and take it as the new segmentation point.

\end{itemize}

In summary, the new segmentation point $\vec p_{j_1}^{m+1}$ can be updated as the following formula
\begin{equation}{\vec p_{j_1}^{m+1}=}
\begin{cases}
\vec p_{j_1}^m + \frac{1}{2}(\vec V_{j-1}^m + \vec V_j^m),\quad\rm{if}\;\vec n_{_\Gamma,_{j-1}}^m\parallel\vec n_{_\Gamma,_j}^m,\cr\\[0.1em]
\vec p_{j_1}^m + \frac{\lambda_{j-1}^m - \lambda_j^m\,R_j^m}{1-|R_j^m|^2}\vec n_{_\Gamma,_{j-1}}^m + \frac{\lambda_{j}^m - \lambda_{j-1}^m\,R_j^m}{1-|R_j^m|^2}\vec n_{_\Gamma,_j}^m,\quad\rm{if}\;\vec n_{_\Gamma,_{j-1}}^m\nparallel\vec n_{_\Gamma,_j}^m,
\end{cases}
\end{equation}
where $R_j^m=\vec n_{_\Gamma,_{j-1}}^m\cdot\vec n_{_\Gamma,_j}^m$. By making use of
\begin{equation}
\Bigl(\vec p_{j_1}^{m+1} - \vec p_{j_1}^m\Bigr)\cdot\vec n_{_\Gamma,_{j-1}}^m  = \lambda_{j-1}^m,\qquad \Bigl(\vec p_{j_1}^{m+1} - \vec p_{j_1}^m\Bigr)\cdot \vec n_{_\Gamma,_{j}}^m = \lambda_{j}^m,
\end{equation}
it is easy to obtain the above formula.

\begin{figure}[!htp]
\centering
\includegraphics[width=0.75\textwidth]{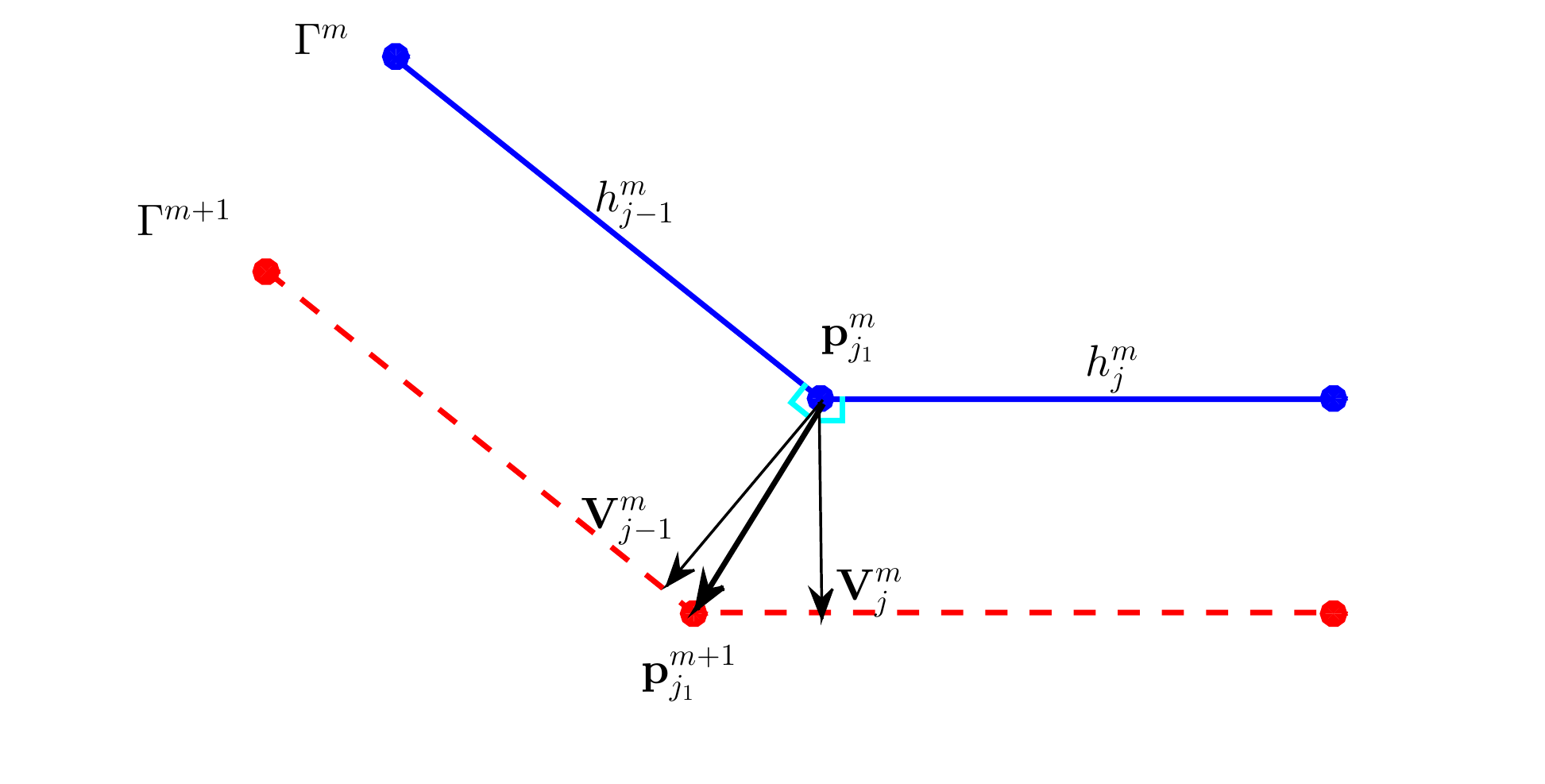}
\caption{The boundary curve is determined by a polygonal line, which can be updated in the following two steps: (1) shift each line segment $h^m_j$ of the curve $\Gamma^m$ along its normal direction by a displacement vector $\vec V_j^m=\lambda_j^m\,\vec n_{_\Gamma,_j}^m$ via the relaxed contact angle condition; (2) calculate the intersection point of the updated adjacent edges, and take it as the segmentation point of the polygonal line $\Gamma^{m+1}$. Specially, if
$\vec n_{_\Gamma,_{j-1}}^m\sslash\vec n_{_\Gamma,_j}^m$, we move the point $\vec p_{j_1}^m$ by a displacement
vector $\frac{1}{2}(\vec V_{j-1}^m + \vec V_j^m)$.}
\label{fig:updateGamma}
\end{figure}

We remark that the above discrete problem results in a linear algebra system which can be efficiently solved via the sparse LU decomposition or GMRES method. Moreover, we have the following theorem for the well-posedness of the proposed discrete scheme.
\begin{thm}[Well-posedness of the PFEM] {\label{thm:AnisoPFEMwellpose}}The above discrete variational problem~\eqref{eqn:fullanisoform1}-\eqref{eqn:fullanisoform2} admits a unique solution (i.e., it is well-posed).
\end{thm}

\begin{proof}
To prove the well-posedness of the PFEM scheme, we need to prove the linear system obtained from~\eqref{eqn:fullanisoform1}-\eqref{eqn:fullanisoform2} has a unique solution. By noting that the moving contact line $\Gamma^{m+1}$ is first updated via the relaxed angle boundary condition in the above PFEM, we can regard it as a Dirichlet type boundary condition for the variational problem \eqref{eqn:fullanisoform1}-\eqref{eqn:fullanisoform2}. It is equivalent to proving the corresponding homogenous linear system has only the zero solution.

Therefore, the well-posedness of the discrete problem~\eqref{eqn:fullanisoform1}-\eqref{eqn:fullanisoform2} is equivalent to that of the following homogeneous linear system: find $\{\vec X^{m+1},~\mu^{m+1}\}\in \{[\mathcal{V}_0^h(S^m)]^3,~V^h(S^m)\}$ such that
\begin{subequations}
\begin{align}
\label{eqn:fullisoform1}
&\big<\vec X^{m+1}\cdot\vec n^m,~\varphi_h\big>_{m}^h + \tau_m\big<\nabla_{_{S}}\mu^{m+1},~\nabla_{_{S}}\varphi_h\big>_{m}^h = 0,\,\forall\varphi_h\in V^h(S^m),\\
&\big<\mu^{m+1},~\vec n^m\cdot\boldsymbol{\omega}_h\big>_{m}^h - \big<\gamma(\vec n^m)\nabla_{_{S}}\vec X^{m+1},~\nabla_{_{S}}\boldsymbol{\omega}_h\big>_{m}^h =0,\,\forall\boldsymbol{\omega}_h\in (\mathcal{V}_0^h(S^m))^3.
\label{eqn:fullisoform2}
\end{align}
\end{subequations}
By choosing the test functions as $\varphi_h = \mu^{m+1},\boldsymbol{\omega}_h = \vec X^{m+1}$, we can immediately obtain
\begin{equation}
\tau_m\big<\nabla_{_{S}}\mu^{m+1},\nabla_{_{S}}\mu^{m+1}\big>_{m}^h + \big<\gamma(\vec n^m)\nabla_{_{S}}\vec X^{m+1},\nabla_{_{S}}\vec X^{m+1}\big>_{m}^h=0.
\label{eqn:immediate}
\end{equation}
By noting $\gamma(\vec n^m)>0$ for all $\vec n^m\in S^2$, we obtain directly $\vec X^{m+1}=\vec 0$ by using the zero boundary condition, and moreover, we have $\mu\equiv \mu^c$ (i.e., a constant). Furthermore, by substituting $\vec X^{m+1}=0$ into~\eqref{eqn:fullisoform2}, we have
\begin{equation}
\mu^c\big<\vec n^m,~\boldsymbol{\omega}^h\big>_{_{m}}^h = 0.
\end{equation}
By choosing $\boldsymbol{\omega}^h=\vec g_j^m\phi_j^m$ with the weighted normal vector $\vec g_j^m$ defined as
\begin{equation}
\vec g_j^m:=\frac{\sum_{D_k^m\in \mathcal{T}_j^m}|D_k^m|\vec n_k^m}{\sum_{D_k^m\in \mathcal{T}_j^m}|D_k^m|}, \quad {\rm with}\quad \mathcal{T}_j^m:=\{D_k^m:\vec q_j^m\in \bar{D}_k^m\},
\label{eqn:wnormal}
\end{equation}
and $\phi_j^m\in V^h(S^m)$ being the nodal basic function at point $\vec q_j^m$,  it immediately yields $\mu^c=0$ by noting the assumption~\eqref{eqn:assumption} and Eq.~\eqref{eqn:masslumpnorm}.

Therefore, the corresponding homogeneous linear system only has the zero solution, which indicates the existence and uniqueness of solution for our PFEM.
\end{proof}

The above proposed PFEM via the $\boldsymbol{\xi}$-vector formulation is an extension to 3D case based on our previous works in 2D~\cite{Jiang18}. The idea behind the variational formulation is by using the decomposition of the Cahn-Hoffman $\boldsymbol{\xi}$-vector into the normal and tangential components~\cite{Jiang18}. In the discrete scheme, the normal component is discretized implicitly, while the tangential components are explicitly discretized. During the practical computation, we need to redistribute mesh points uniformly in 2D according to the arc-length for the polygonal boundary line in each time step; similarly, we also use the mesh redistribution algorithm discussed in~\cite{Bansch05} to prevent the mesh distortion for
the triangular surface mesh.

Furthermore, since $S^{m+1}:=\vec X^{m+1}(S^m)$ is assumed to be parameterized over $S^m$, the operator $\nabla_{_{S}}$ can then be very easily numerically calculated. More precisely, consider the triangular surface $D_j^m$ with vertices $\{\vec q_{j_1}^m,~\vec q_{j_2}^m,~\vec q_{j_3}^m\}$ ordered in the anti-clockwise direction, we then have
\begin{equation}
\nabla_{_{S}} B_{j_1}(S^m)\Big|_{D_j^m} = \frac{(\vec q_{j_3}^m-\vec q_{j_2}^m)\times\vec n_j^m}{2|D_j^m|},
\end{equation}
where $B_{j_1}\in V^h(S^m)$ is the nodal basis function defined at point $\vec q_{j_1}^m$. Similarly, we can easily obtain $\nabla_{_S}B_{j_2}$ and $\nabla_{_S}B_{j_3}$. Therefore, for any piecewise linear function $\phi\in V^h(S^m)$, we can have
\begin{equation}
 \nabla_{_{S}}\phi\Big|_{D_j^m} = \sum_{i=1}^3\phi(\vec q_{j_i}^m)\nabla_{_{S}}B_{j_i}.
\end{equation}

\section{Numerical results}

In this section, we implement the proposed PFEM, show some equilibrium convergence results, and perform lots of numerical simulations to demonstrate the efficiency and accuracy of the proposed scheme. In the following simulations, we use the uniform
time step, i.e., $\tau=\tau_m$, $m=0,1,2,\ldots$.

\subsection{Equilibrium convergence}

The mathematical description of the equilibrium shape has been fully investigated in~\cite{Bao18a}. Here, we present some numerical equilibrium convergence results by solving the kinetic sharp-interface model via the proposed PFEM scheme.

From the relaxed contact angle boundary condition~\eqref{eqn:boundcon1}, we know that the contact line mobility $\eta$ precisely controls the relaxation rate of the contact angle towards its equilibrium state. The large $\eta$ will accelerate the relaxation process~\cite{Wang15,Jiang12,Huang19b}. Here, we numerically investigate the effect of $\eta$ on the evolution of the dynamic contact angles.  We numerically define the following average contact angle $\bar{\theta}^m$ as the indicator,
\begin{equation}
\bar{\theta}^m = \frac{1}{N_c}\sum_{j=1}^{N_c} \arccos(\vec c^m_{_{\Gamma,j}}\cdot\vec n^m_{_{\Gamma,j}}),
\label{eqn:meantheta}
\end{equation}
where $\vec n^m_{_{\Gamma,j}}$ and $\vec c^m_{_{\Gamma,j}}$ are the unit normal and co-normal vectors defined on the $j$-th line segment $h_j^m$ of the boundary curve $\Gamma^m$.
\begin{figure}[!htp]
\centering
\includegraphics[width=0.95\textwidth]{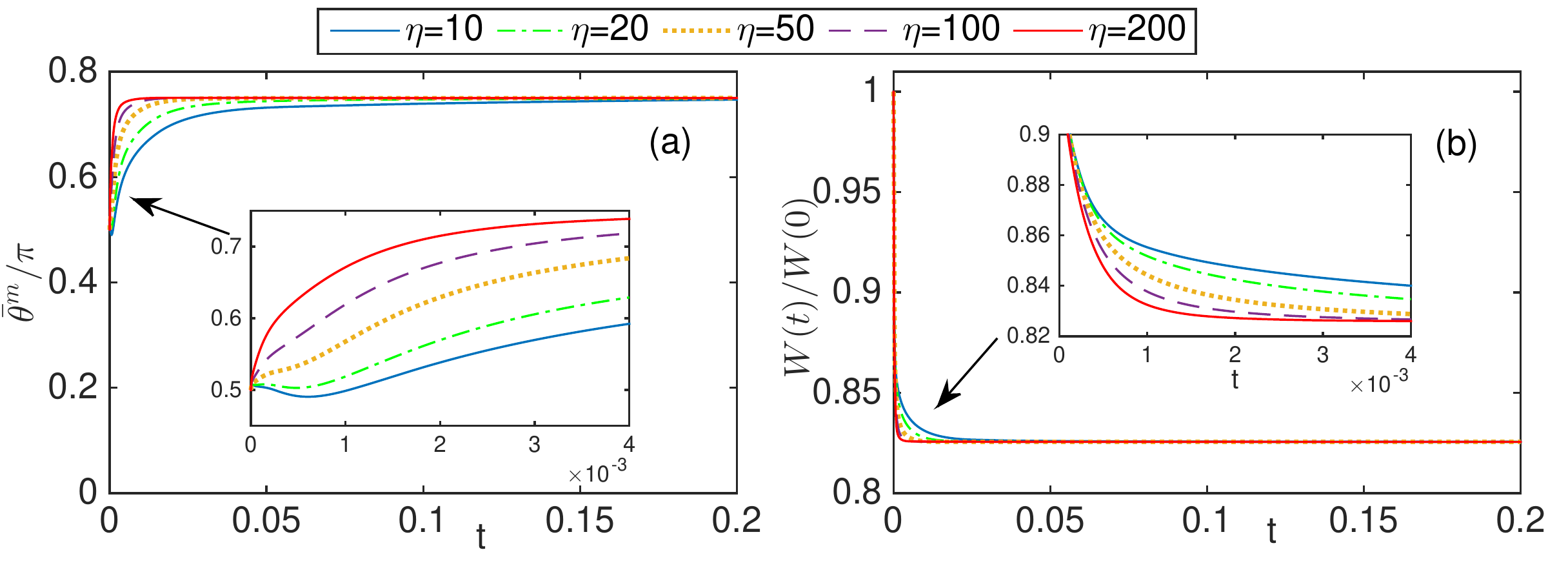}
\caption{(a) The temporal evolution of the average contact angle $\bar{\theta}^m$ defined in \eqref{eqn:meantheta}; (b) the temporal evolution of the normalized energy $W(t)/W(0)$ for different choices of mobility, where the initial shape of the island film with isotropic surface energy is chosen as a unit cube, and the computational parameters are chosen as $\sigma=\cos(3\pi/4)$.}
\label{fig:smallisland111}
\end{figure}

Fig.~\ref{fig:smallisland111} shows the temporal evolution of $\bar{\theta}^m$ and the normalized energy $W(t)/W(0)$ under different choices of the contact line mobility $\eta$. The initial shape of the island film is chosen as a unit cube, and $\sigma=\cos(3\pi/4)$. From the figure, we can observe that the larger mobility $\eta$ will accelerate the process of relaxation such that the contact angles evolve faster towards its equilibrium contact angle $3\pi/4$. As shown in Fig.~\ref{fig:smallisland111}, the energy decays faster for larger mobility, but finally, it converges to the same equilibrium state. It indicates that the equilibrium contact angle, as well as the equilibrium shape, is independent of the choice of the contact line mobility $\eta$. Meanwhile, the total volume loss (not shown here) of the island film is always below 0.5\% during the numerical simulations. In the following numerical simulations, the contact line mobility is chosen to be very large (e.g., $\eta=100$). This choice of $\eta$ will result in a very quick convergence to the equilibrium contact angle (defined in \eqref{eqn:staticangle}). The detailed investigation of the influence of the parameter $\eta$ on the solid-state dewetting evolution process and equilibrium shapes was performed in 2D~\cite{Wang15}.

\begin{figure}[!htp]
\centering
\includegraphics[width=0.75\textwidth]{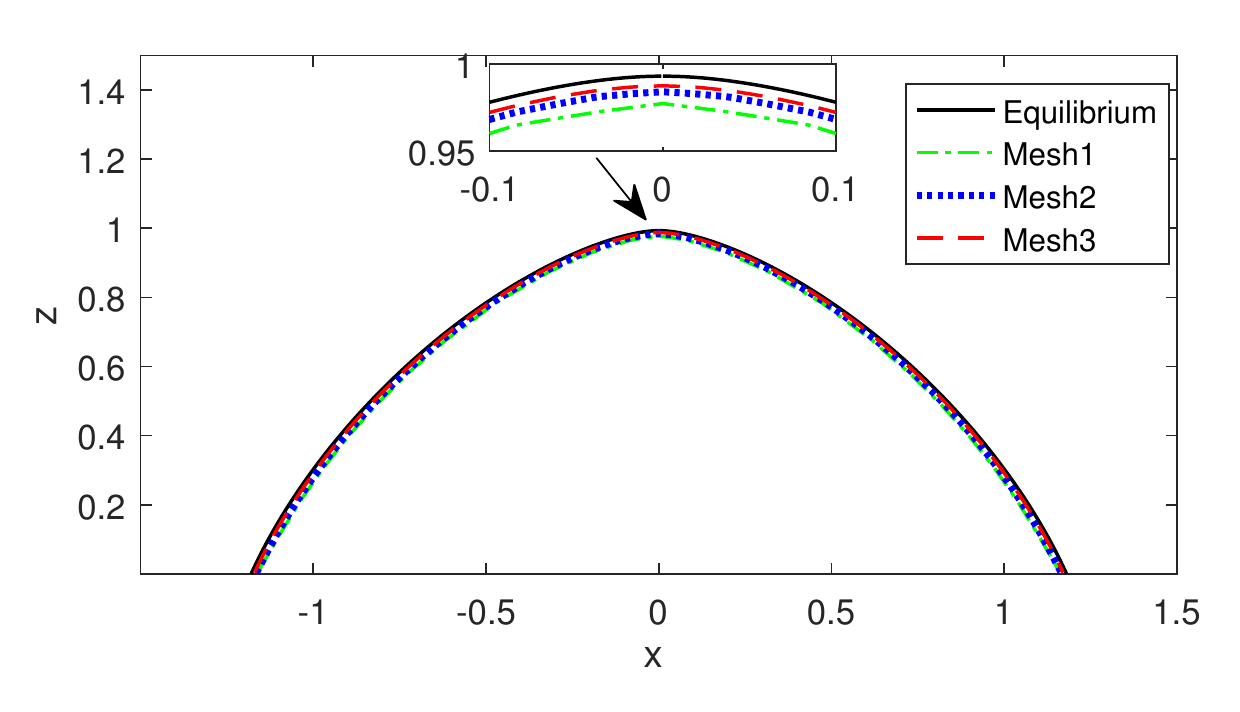}
\caption{Comparisons of the cross-section profiles along the $x$-direction of the numerical equilibrium shapes under different meshes with its theoretical equilibrium shape, where the initial shape is chosen as a $(1,2,1)$ cuboid, the surface energy $\gamma(\vec n)=1+0.25(n_1^4+n_2^4+n_3^4)$, and $\sigma=\cos(15\pi/36)$. The theoretical equilibrium shape (black line) is constructed by the Winterbottom construction~\cite{Winterbottom67,Bao17b}.}
\label{fig:Equcon}
\end{figure}

We next show a convergence result between the numerical equilibrium shapes by solving the proposed sharp-interface model
and its theoretical equilibrium shape. Fig.~\ref{fig:Equcon} depicts equilibrium shapes under different mesh sizes, where $\sigma=\cos(15\pi/36)$, $\gamma(\vec n)=1+0.25(n_1^4+n_2^4+n_3^4)$. The initial shape is
chosen as a $(1,2,1)$ cuboid, then we numerically evolve it until the equilibrium state by using different meshes, which are given by a set of small isosceles right triangles. If we define the mesh size indicator $h$ as the length of the hypotenuse of the isosceles right triangle, then ``Mesh $1$'' represents the initial mesh with $h=h_0=0.125$, and the time step is chosen as $\tau =\tau_0=0.00125$ for numerical computation. Meanwhile, the time step for ``Mesh $2$'' ($h=h_0/2$) and ``Mesh $3$'' ($h=h_0/4$) are chosen as $\tau = \tau_0/4$ and $\tau=\tau_0/16$, respectively. For a better comparison, we plot the cross-section profiles along the
$x$-direction for the numerical equilibrium shapes and the theoretical equilibrium shape. As shown in Fig.~\ref{fig:Equcon}, we can clearly observe that as the computational mesh size gradually decreases, the numerical equilibrium shapes uniformly converge to the theoretical equilibrium shape.

\subsection{For isotropic case}
\begin{figure}[!htp]
\centering
\includegraphics[width=0.98\textwidth]{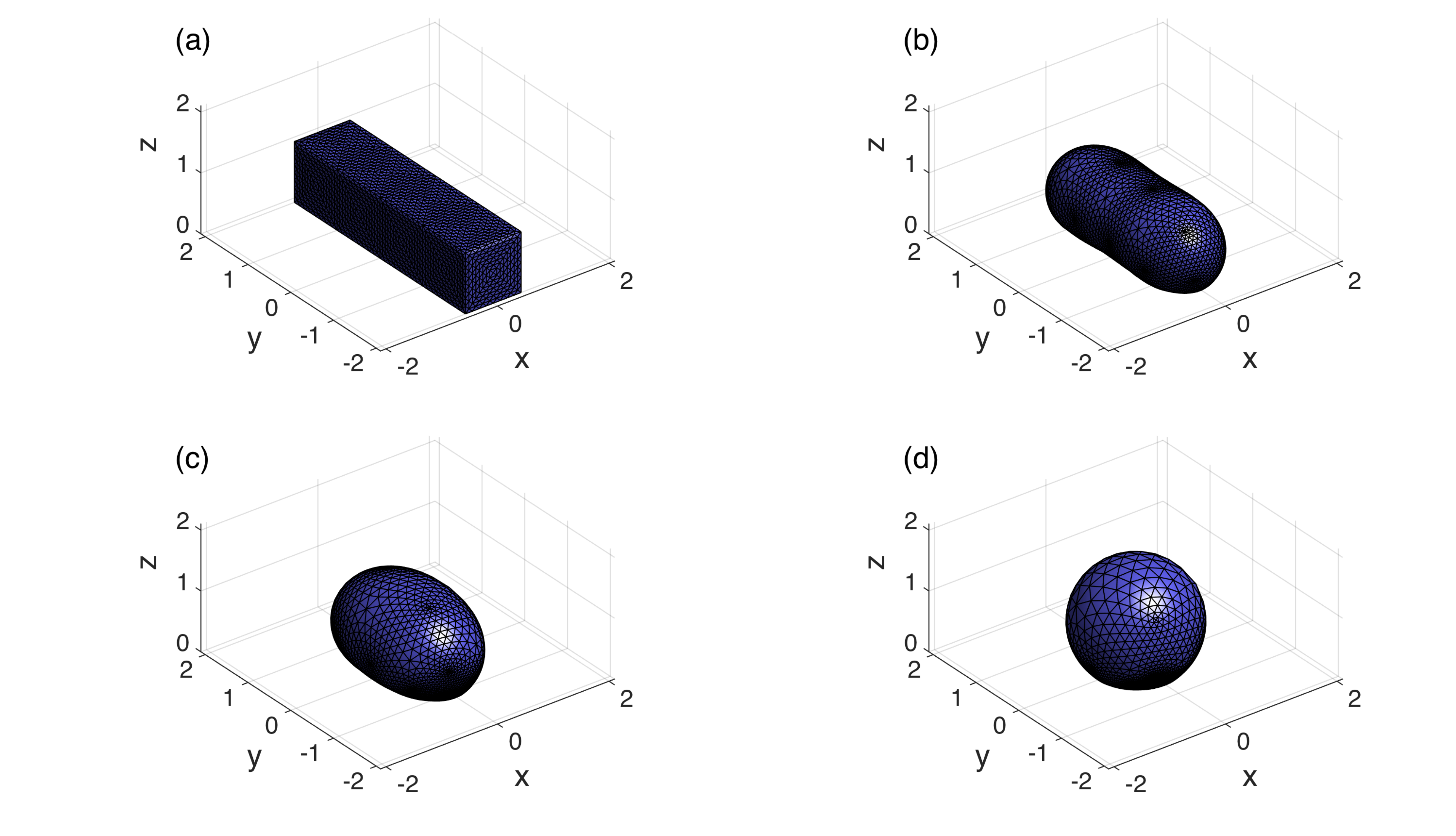}
\caption{Several snapshots in the evolution of an initial $(1,4,1)$ cuboid island towards its equilibrium shape: (a) $t=0$; (b) $t=0.10$; (c) $t=0.20$; (d) $t=1.94$, where the material constant is chosen as $\sigma=\cos(3\pi/4)$.}
\label{fig:141islands}
\end{figure}
We first focus on the isotropic surface energy case, i.e., $\gamma(\vec n)\equiv 1$. We start with a numerical example by initially choosing a small cuboid island with $(1,4,1)$ representing its width, length and height, and the material constant is chosen as $\sigma=\cos(3\pi/4)$.  The cuboid is initially almost uniformly discretized into 3584 small isosceles right triangles with total 1833 vertices and 80 vertices on the boundary curve. The time step is chosen uniformly as $\tau_m=2\times 10^{-4}$. As is shown in Fig.~\ref{fig:141islands}, it depicts several snapshots of the triangular surface mesh of the island towards its equilibrium shape. We can clearly observe that the sharp corner of the island gradually disappears and becomes smoother and smoother, and finally, the island evolves into a perfect spherical shape which is truncated by the flat substrate.

\begin{figure}[!htp]
\centering
\includegraphics[width=0.75\textwidth]{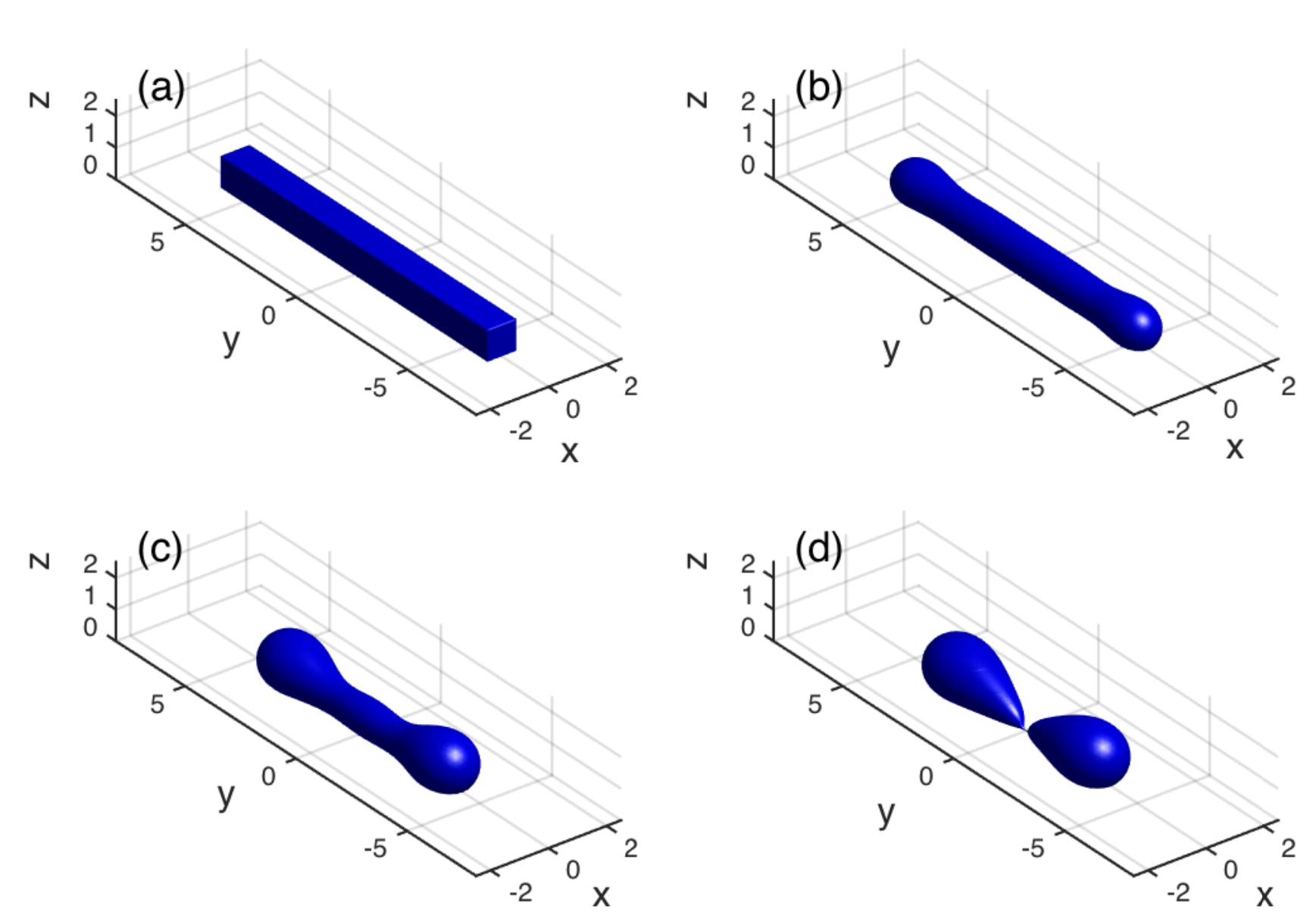}
\caption{Several snapshots in the evolution of an initial $(1,12,1)$ cuboid island until its pinch-off: (a) $t=0$; (b) $t=0.01$; (c) $t=0.75$; (d) $t=1.07$, where the material constant is chosen as $\sigma=\cos(3\pi/4)$. The initial surface mesh consists of 9728 triangles and 4969 vertices with 208 vertices on the boundary, and the time step is uniformly chosen as $\tau_m=10^{-4}$.}
\label{fig:1121islands}
\end{figure}
\begin{figure}[!htp]
\centering
\includegraphics[width=0.75\textwidth]{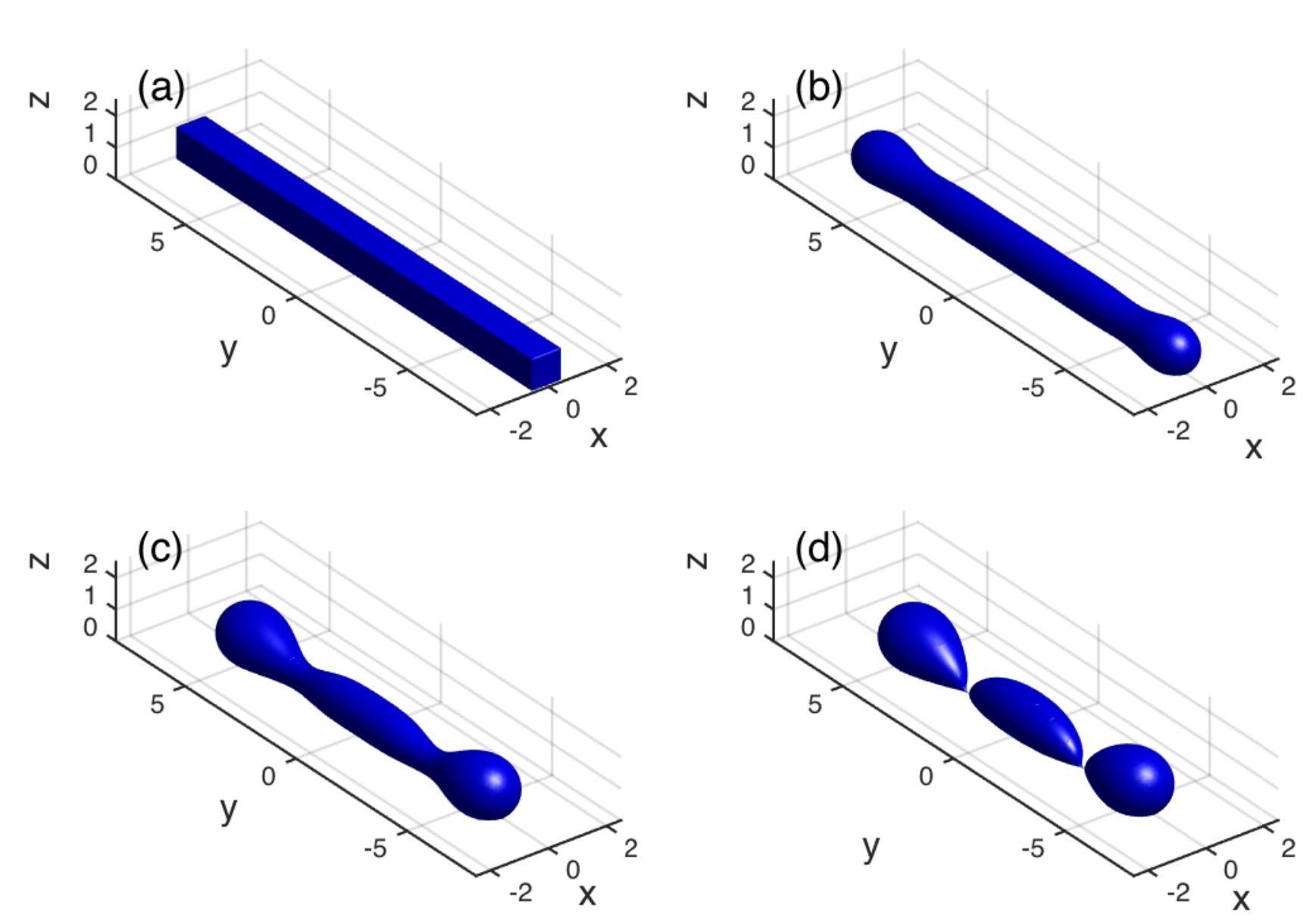}
\caption{Several snapshots in the evolution of an initial $(1,16,1)$ cuboid island until its pinch-off: (a) $t=0$; (b) $t=0.20$; (c) $t=0.90$; (d) $t=1.14$, where the material constant is chosen as $\sigma=\cos(3\pi/4)$. The initial surface mesh consists of 12800 triangles and 6537 vertices with 272 vertices on the boundary, and the time step is uniformly chosen as $\tau_m=10^{-4}$.}
\label{fig:1161islands}
\end{figure}

\begin{figure}[!htp]
\centering
\includegraphics[width=0.75\textwidth]{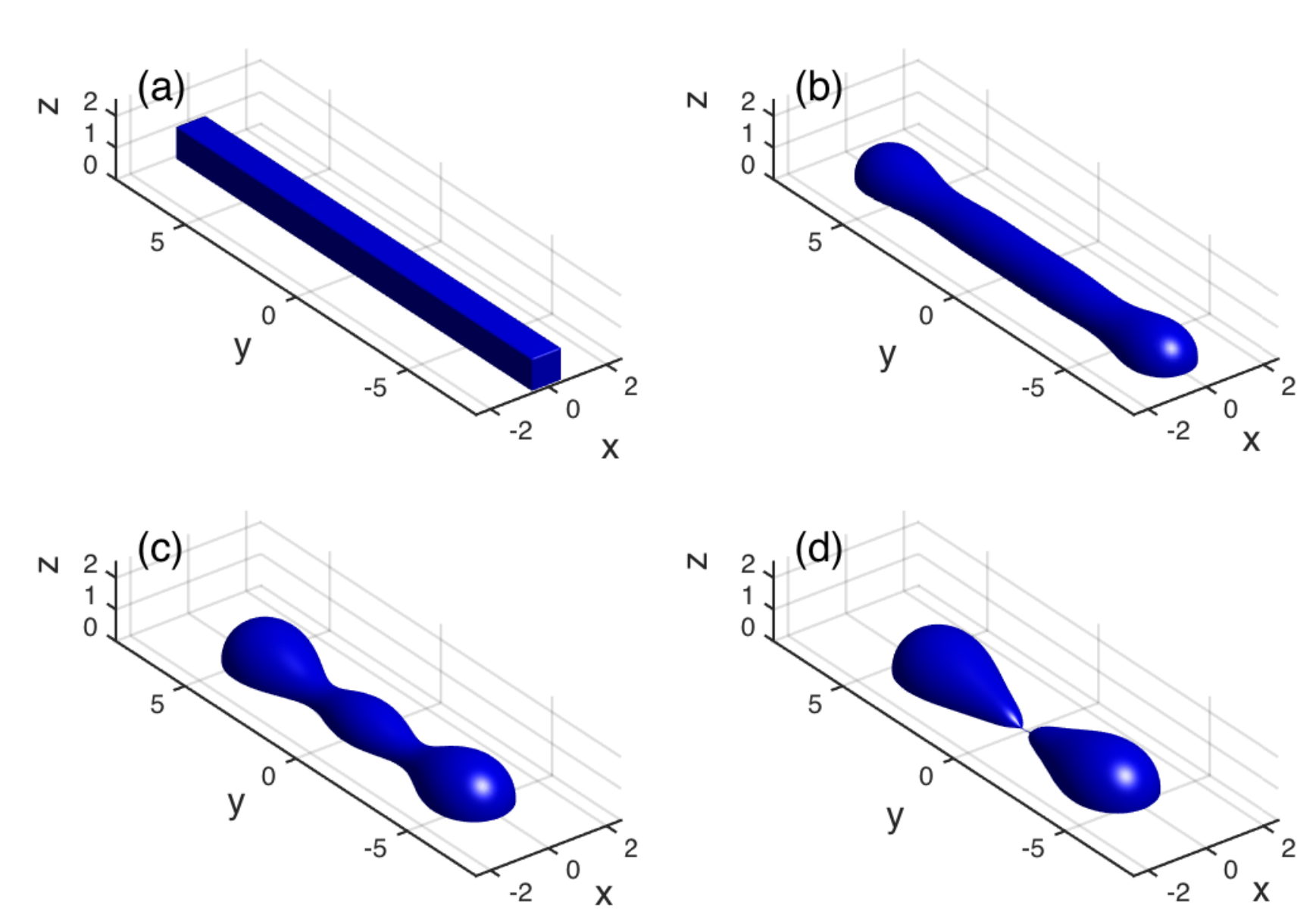}
\caption{Several snapshots in the evolution of an initial $(1,16,1)$ cuboid island until its pinch-off: (a) $t=0$; (b) $t=0.50$; (c) $t=2.00$; (d) $t=3.40$, where $\sigma=\cos(\pi/2)$. The initial surface mesh consists of 12800 triangles and 6537 vertices with 272 vertices on the boundary, and the time step is uniformly chosen as $\tau_m=10^{-4}$.}
\label{fig:1161islandsig}
\end{figure}
In general, a short island film tends to form a single spherical shape in order to arrive at its lowest energy state, while a long island film will pinch off and agglomerate into pieces of small isolated islands before it reaches at a single spherical shape. This pinch-off phenomenon has often been identified as the Rayleigh-like instability~\cite{Kim15,Rayleigh78} governed by surface diffusion. To study this particular phenomenon for solid-state dewetting problems, we perform a lot of numerical simulations with different initial islands given by different lengths of $(1,L,1)$ cuboids. As shown in Fig.~\ref{fig:1121islands} and Fig.~\ref{fig:1161islands}, for an initial $(1,12,1)$ cuboid island with material constant $\sigma=\cos(3\pi/4)$,  we can observe that the island evolves and breaks up into $2$ small isolated islands, and an initial $(1,16,1)$ cuboid island could break up into $3$ pieces of islands. Furthermore, by changing $\sigma=\cos(\pi/2)$, we observe that an initial $(1,16,1)$ cuboid island only breaks up into two small isolated islands (cf.~Fig.~\ref{fig:1161islandsig}). This indicates that when $\sigma$ increases, a cuboid island will become more difficult to pinch off.

\begin{figure}[!htp]
\centering
\includegraphics[width=0.85\textwidth]{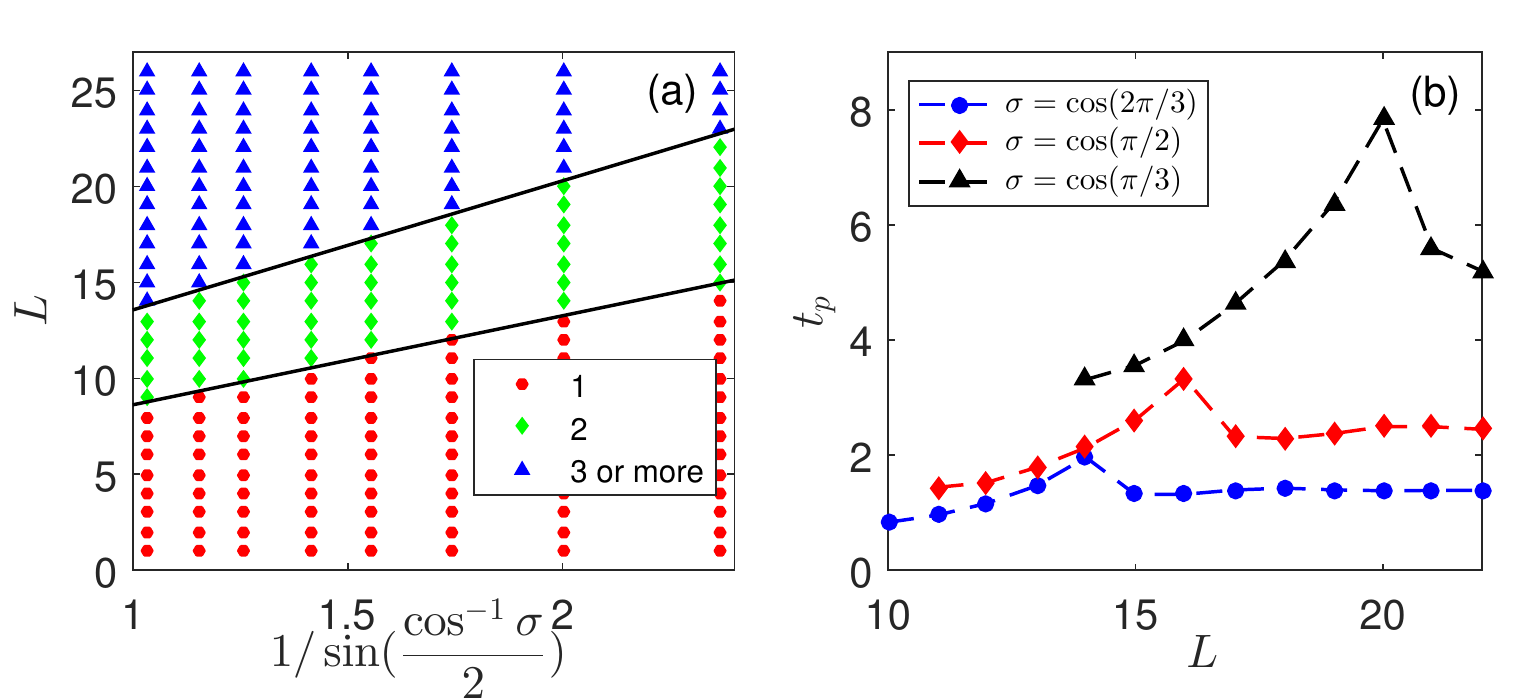}
\caption{(a) The number of islands formed from an initial $(1,L,1)$ cuboid island with material constant $\sigma$, where the 1-2 islands and 2-3 islands boundaries (solid lines) are linear curve fittings to our numerical simulations given by $L = 3.98 + 4.64/{\sin(\arccos\sigma/2)}$ and $L=6.84 + 6.73/{\sin(\arccos\sigma/2)}$; (b) The first pinch-off time $t_p$ for an initial $(1,L,1)$ cuboid island under three different material constant $\sigma$. }
\label{fig:PhaseL}
\end{figure}

From the above numerical simulations, we observe that there exist two critical lengths $L_1, L_2$ such that when $L_1< L< L_2$, an initial $(1,~L,~1)$ cuboid island will break up into $2$ small isolated particles; and when $L > L_2$, the cuboid island will break up into $3$ or more. Furthermore, we also observe that these two critical lengths are highly dependent on the material constant $\sigma$. By performing ample numerical simulations, as shown in Fig.~\ref{fig:PhaseL}(a), we plot the phase diagram for the numbers of islands formed from an initial $(1, L,1)$ cuboid island under different material constants $\sigma$. From the figure, we can observe that the critical lengths $L_1$ and $L_2$ both exhibit the reciprocal linear relationship with the variable $\sin(\arccos{\sigma}/2)$. We note that several similar relationships have also been observed and reported for the solid-state dewetting in 2D~\cite{Wang15,Dornel06}. Moreover, we also plot the first pinch-off time $t_p$ for an initial $(1,L,1)$ cuboid island under three different material constants, i.e., $\sigma=\cos(\pi/3),\cos(\pi/2),\cos(2\pi/3)$.
As shown in Fig.~\ref{fig:PhaseL}(b), we can observe that when $L$ increases, the first pinch-off time $t_p$ first increases
quickly to a maximum value, then decreases slowly to a constant. This is certainly reasonable because for an infinitely long
$(1,L,1)$ cuboid island, its first pinch-off time should be a constant which is only dependent of $\sigma$.

Motivated by recent experiments by Thompson's group~\cite{Thompson12,Ye11b}, we next numerically investigate morphology evolutions for island films initially with some special geometries, such as the cross shape and square-ring shape. In the following simulations, the height of the initial island film is always chosen to be $1$, and the material constant is fixed at $\sigma=\cos(3\pi/4)$, unless otherwise stated.

\begin{figure}[!htp]
\centering
\includegraphics[width=0.75\textwidth]{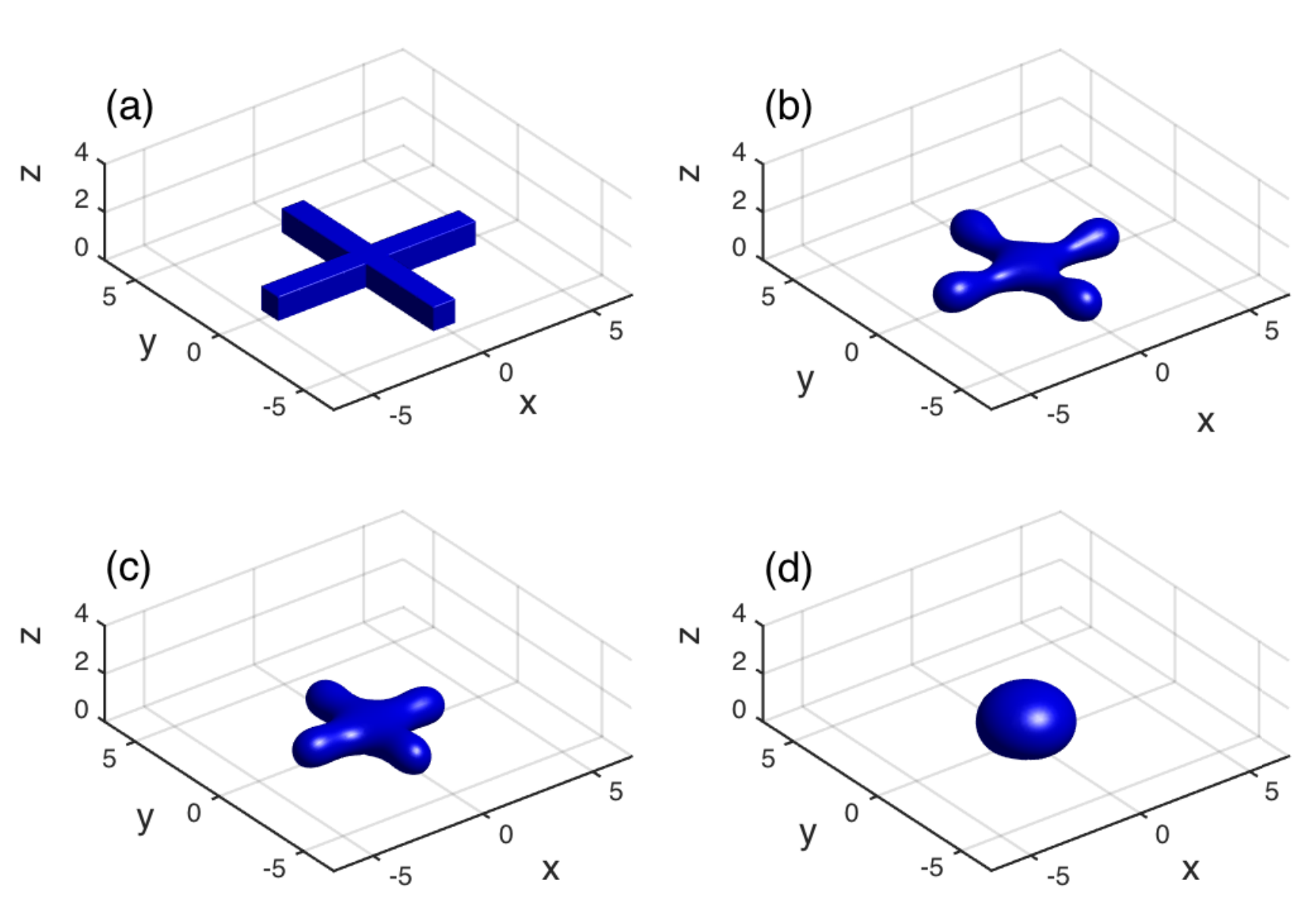}
\caption{Several snapshots in the evolution of an initially cross-shaped island towards its equilibrium, where the initial island consists of four (1,4,1) cuboids forming the limbs and one (1,1,1) cube sitting in the centre: (a) $t=0$; (b) $t=0.15$; (c) $t=0.50$; (d) $t=1.40$. The initial surface mesh consists of 13568 triangles and 6929 vertices with 289 vertices on the boundary, and the time step is uniformly chosen as $\tau_m=10^{-4}$.}
\label{fig:CrossShape1}
\end{figure}

\begin{figure}[!htp]
\centering
\includegraphics[width=0.75\textwidth]{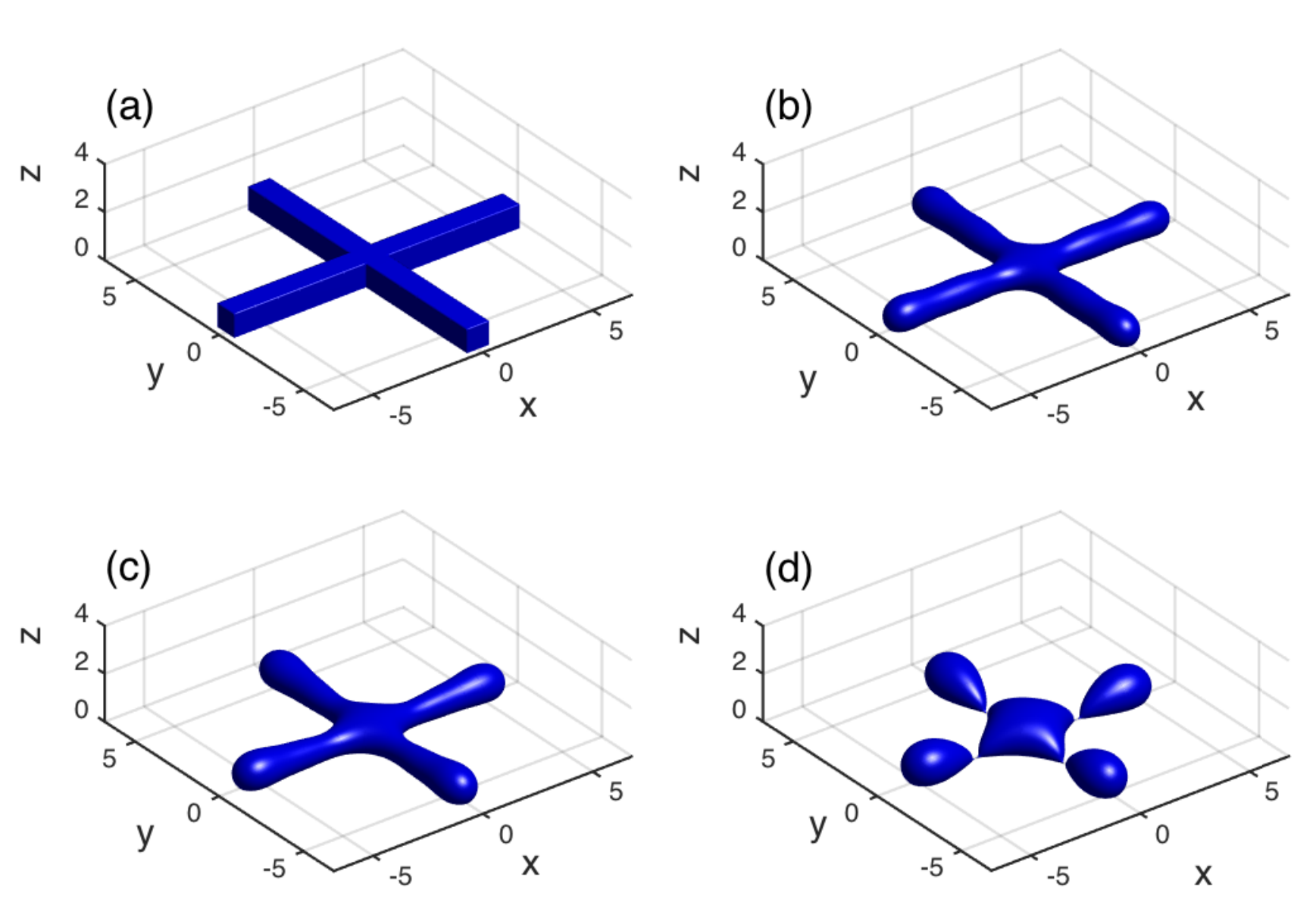}
\caption{Several snapshots in the evolution of an initially cross-shaped island before its pinch-off, where the initial island consists of four (1,6,1) cuboids forming the limbs and one (1,1,1) cube sitting in the centre: (a) $t=0$; (b) $t=0.05$; (c) $t=0.15$; (d) $t=0.386$. The initial surface mesh consists of 19712 triangles and 10065 vertices with 416 vertices on the boundary, and the time step is uniformly chosen as $\tau_m=10^{-4}$.}
\label{fig:CrossShape2}
\end{figure}

To compare evolution process with the recent experiments~\cite{Thompson12,Ye11b}, we first choose the initial geometry of the island film as a unit cube plus four equal limbs which are given by four $(1,L,1)$ cuboids (shown in Fig.~\ref{fig:CrossShape1}(a)). We test two numerical examples with length parameters $L=4$ and $L=6$. As can been seen in Fig.~\ref{fig:CrossShape1}, when the limbs are chosen to be shorter (i.e., $L=4$), we observe that the four limbs of the islands shrinks, and then the cross-shaped island eventually evolves into a single island with spherical geometry as its equilibrium shape. However, when the four limbs are chosen to be longer (i.e., $L=6$), its kinetic evolution of the island could be quite different. As depicted in Fig.~\ref{fig:CrossShape2}, instead of eventually forming a single spherical island, the cross-shaped island undergoes the pinch-off at somewhere and finally breaks up into five small isolated solid particles.

\begin{figure}[!htp]
\centering
\includegraphics[width=0.75\textwidth]{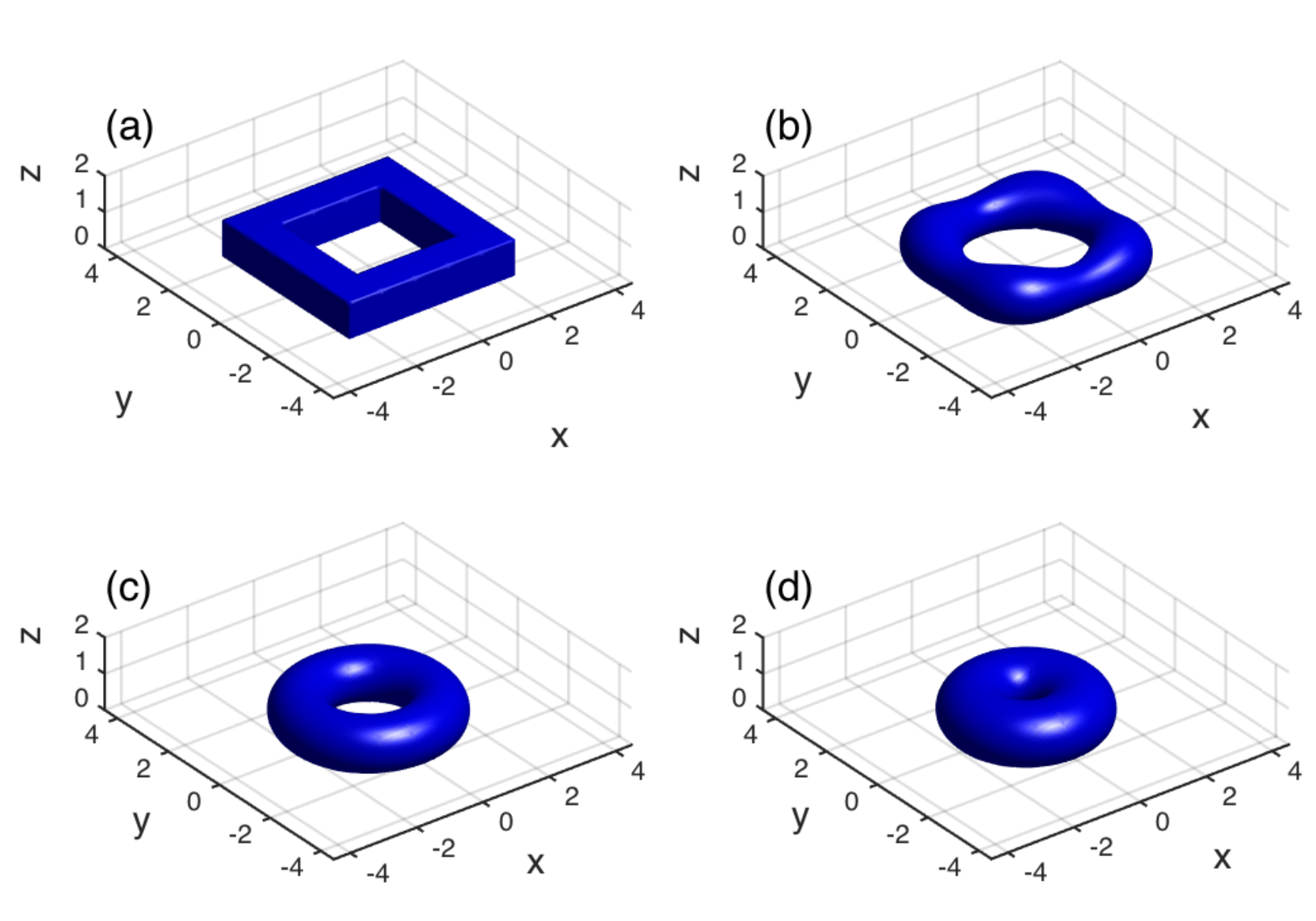}
\caption{Several snapshots in the evolution of an initial island of square-ring obtained from a $(5,5,1)$ cuboid by cutting out a $(3,3,1)$ cuboid from center: (a) $t=0$; (b) $t=0.15$; (c) $t=1.00$; (d) $t=1.50$. The initial surface mesh consists of 12288 triangles and 6272 vertices with 96 and 160 vertices for the inner and outer contact lines, respectively, and the time step is uniformly chosen as $\tau_m=5\times 10^{-4}$.}
\label{fig:hollow12}
\end{figure}

\begin{figure}[!htp]
\centering
\includegraphics[width=0.95\textwidth]{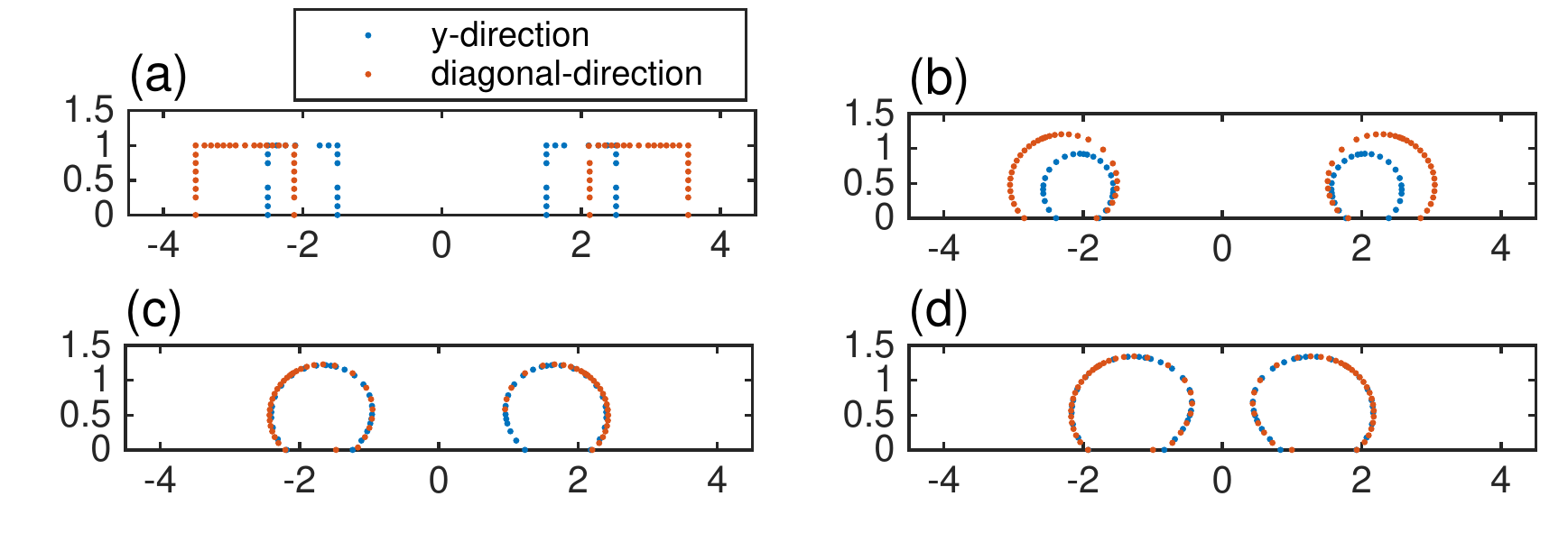}
\caption{The corresponding cross-section profiles of the island geometry in the evolution shown in Fig.~\ref{fig:hollow12}: (a) $t=0$; (b) $t=0.15$; (c) $t=1.00$; (d) $t=1.50$.}
\label{fig:hollow121r}
\end{figure}

We next consider the evolution of an island film which is initially chosen as a square-ring shape.
First, we choose an initial `fat' square-ring island, which is given by a $(5,5,1)$ cuboid by cutting out a $(3,3,1)$ cuboid from center (shown in Fig.~\ref{fig:hollow12}). Its geometry evolution as well the cross-section profile of the island are shown in Fig.~\ref{fig:hollow12} and Fig.~\ref{fig:hollow121r}, respectively. From these figures, we clearly observe that the square-ring island quickly evolves into a ring-like shape with different thickness along different cross-section directions
(see Fig.~\ref{fig:hollow121r}(b)). Subsequently, as time evolves, this ring-like shape evolves into a perfect toroidal shape (i.e., its thickness is the same along each cross-section direction) (see Fig.~\ref{fig:hollow121r}(c)), then the toroidal
island shrinks towards the center in order to reduce the total free energy.

\begin{figure}[!htp]
\centering
\includegraphics[width=0.75\textwidth]{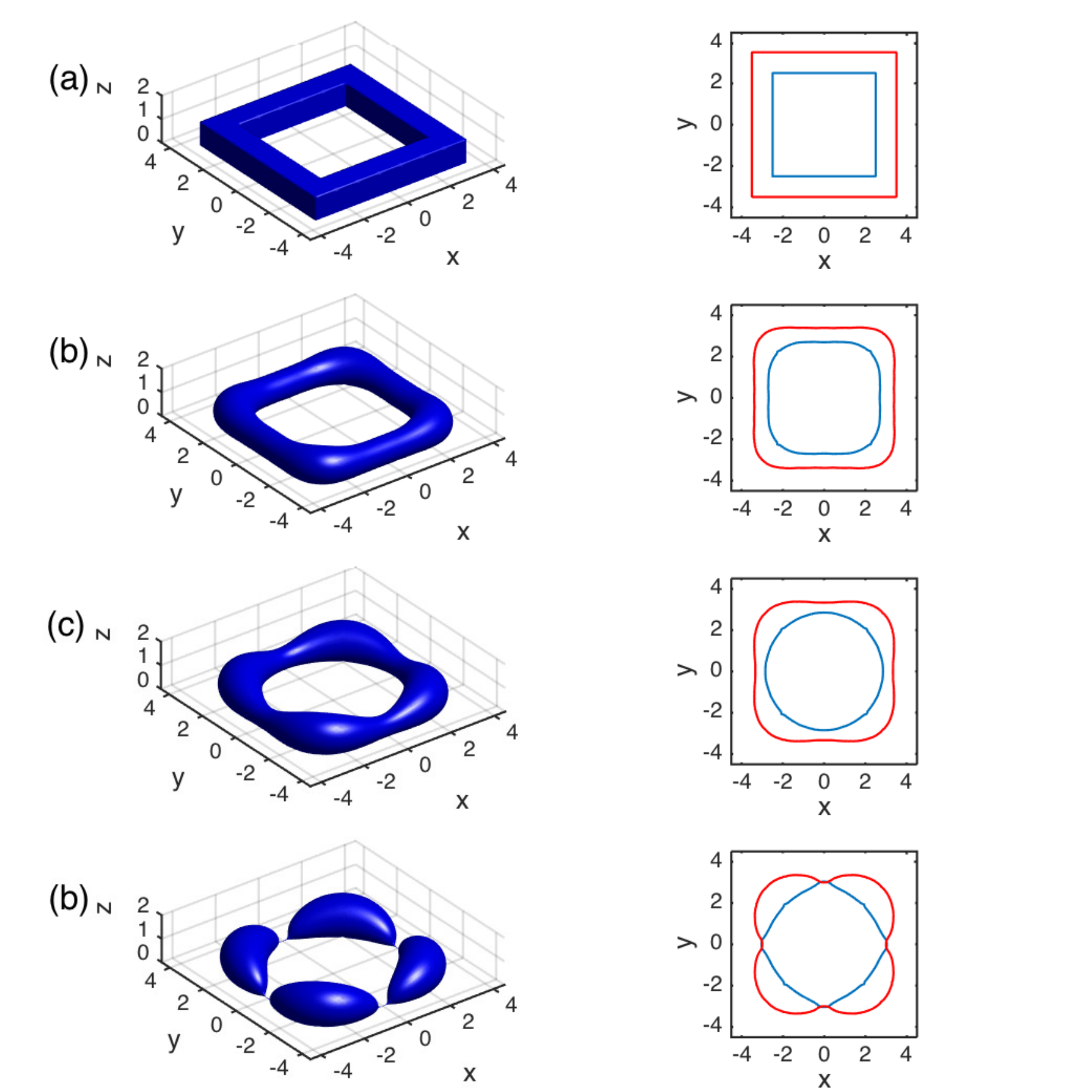}
\caption{Several snapshots in the evolution of an initial square-ring island obtained from a $(7,7,1)$ cuboid by cutting out a $(5,5,1)$ cuboid from center: (a) $t=0$; (b) $t=0.15$; (c) $t=0.40$; (d) $t=0.61$.}
\label{fig:hollow13}
\end{figure}

Furthermore, if we choose an initial `thin' square-ring island (i.e., enlarge the length of outer edge of the island, while fixing the inner-width of square-ring island still as $2$), the pinch-off events will occur as expected due to Rayleigh-like instability, as shown in Fig.~\ref{fig:hollow13} and Fig.~\ref{fig:hollow15}. Fig.~\ref{fig:hollow13} depicts the morphology
evolution and contact line migration (including inner and outer contact lines), when the length of the outer edge is chosen
as $7$. From this figure, we clearly see that the inner-width of the island becomes very quickly wavy along its different azimuthal directions; and as time evolves, the place where its inner-width is thick becomes thicker and thicker, while the place where it is thin becomes thinner and thinner; finally, when the width of thin place approaches to zero, the pinch-off events will happen such that it breaks up into $4$ pieces of small particles. On the other hand, if we continue to enlarge the length of outer edge (e.g., choose it as $12$), as shown in Fig.~\ref{fig:hollow15}, we can observe that the square-ring island will finally split into more pieces of small islands.

From the above numerical simulations, we can observe that the Rayleigh-like instability in the azimuthal direction and the shrinking instability in the radial direction are competing with each other to determine the solid-state dewetting evolution of a square-ring island. This is a competition between the two time scales: one for toroid shrinkage towards its center and the other for neck pinch-off along the azimuthal direction. When the square-ring island is very thin (shown in Fig.~\ref{fig:hollow13} and Fig.~\ref{fig:hollow15}), the Rayleigh-like instability dominates its kinetic evolution, and makes the island break up into small isolated pieces of particles; when it is very fat (shown in Fig.~\ref{fig:hollow12}),
the shrinking instability dominates the evolution, and make it shrink towards the center. The shrinking instability for a toroidal island on a substrate has been studied in~\cite{Jiang19,Zhao19} under the assumption of axis-symmetric geometry.
But it is still an open problem about quantitatively studying the competition effect by a simultaneous consideration of the shrinking instability and Rayleigh-like instability. Our proposed approach could offer a good candidate for exploring
this problem.

\begin{figure}[!htp]
\centering
\includegraphics[width=0.75\textwidth]{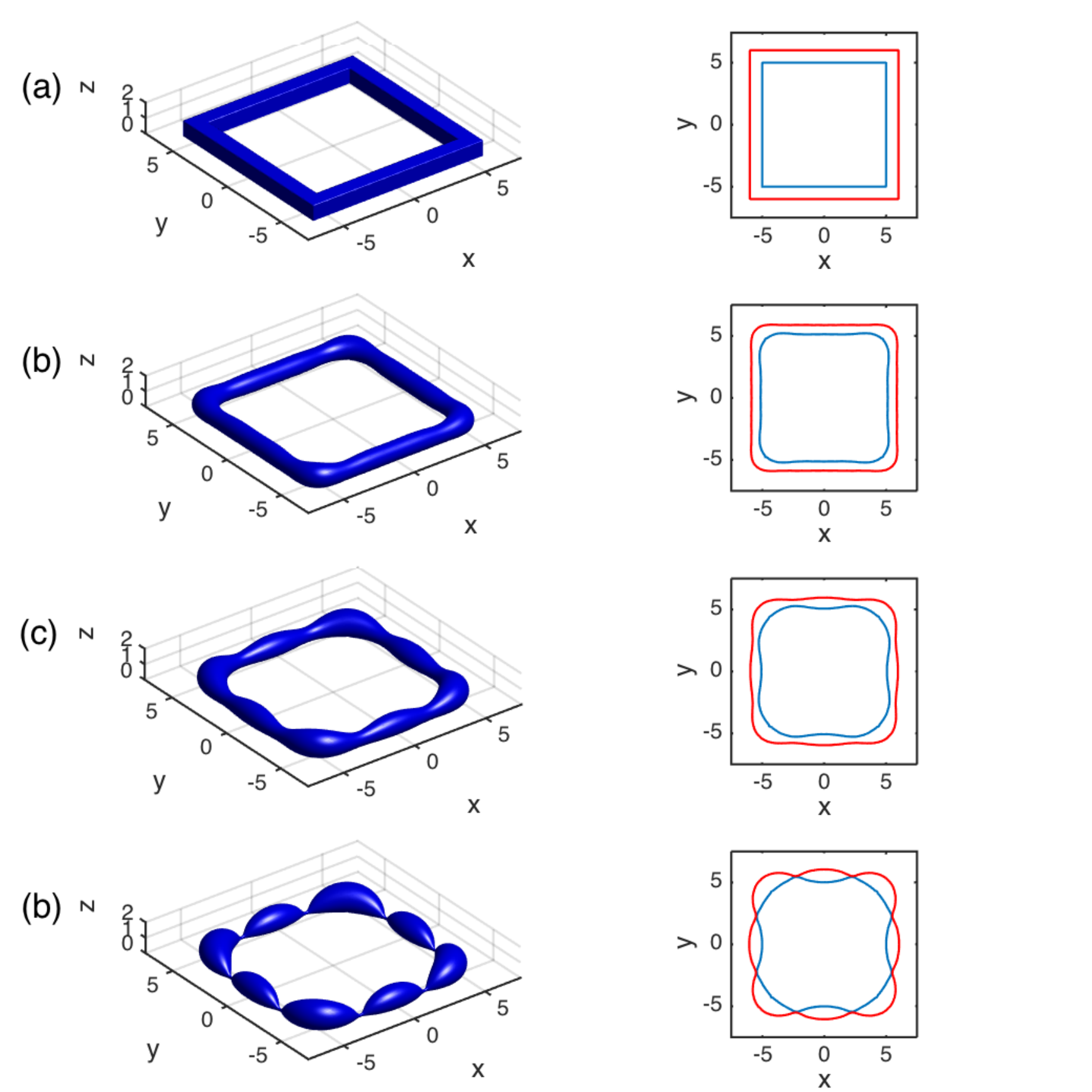}
\caption{Several snapshots in the evolution of an initial square-ring island obtained from a $(12,12,1)$ cuboid by cutting out a $(10,10,1)$ cuboid from center: (a) $t=0$; (b) $t=0.15$; (c) $t=0.70$; (d) $t=1.00$. }
\label{fig:hollow15}
\end{figure}

\subsection{For weakly anisotropic case}
\begin{figure}[!htp]
\centering
\includegraphics[width=0.75\textwidth]{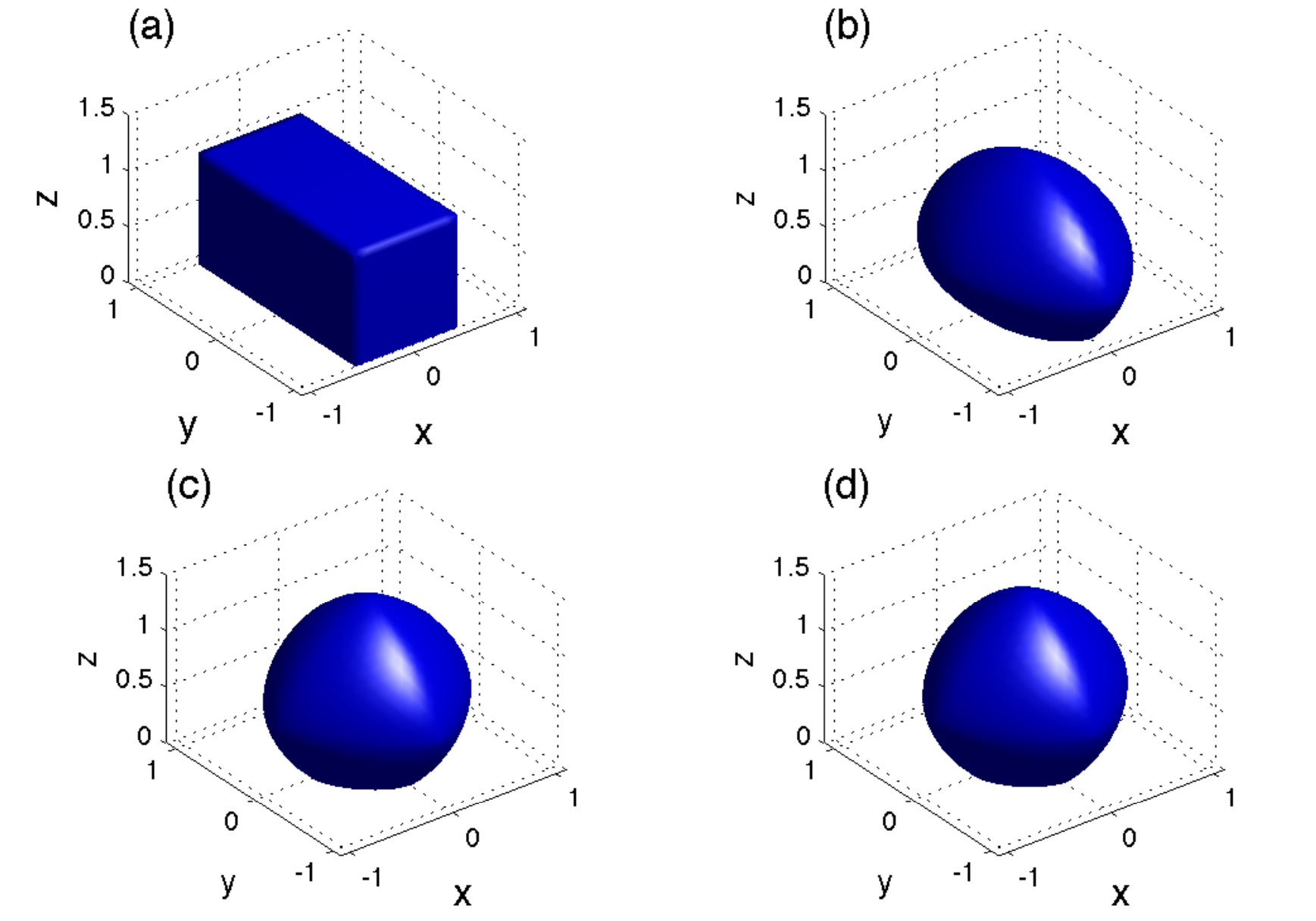}
\caption{Several snapshots in the evolution of an initially $(1,2,1)$ cuboid island towards its equilibrium under the cubic anisotropy with $a=0.3$: (a) $t=0$; (b) $t=0.02$; (c) $t=0.10$; (d) $t=0.21$, where $\sigma=\cos(5\pi/6)$, and the initial surface mesh consists of 2048 triangles and 1049 vertices with 48 vertices on the boundary, and the time step is uniformly chosen as $\tau_m=10^{-4}$.}
\label{fig:weak121}
\end{figure}

In this subsection, we perform some numerical simulations to investigate solid-state dewetting of thin films with anisotropic surface energies. We first focus on the following cubic anisotropy:
\begin{equation}
\gamma(\vec n) = 1 + a[n_1^4 + n_2^4 + n_3 ^4],\qquad 0\leq a<\frac{1}{3},
\end{equation}
where $a$ represents the degree of the anisotropy.

We start the numerical experiment for an initial $(1,2,1)$ cuboid island. The surface energy is chosen as the cubic anisotropy with $a=0.3$, and the material constant is chosen as $\sigma=\cos(5\pi/6)$. Several snapshots of the morphology evolution of the island film are shown in Fig.~\ref{fig:weak121}. From the figure, we can observe that the island film evolves towards a non-spherical shape with `sharp' corners.

\begin{figure}[!htp]
\centering
\includegraphics[width=0.75\textwidth]{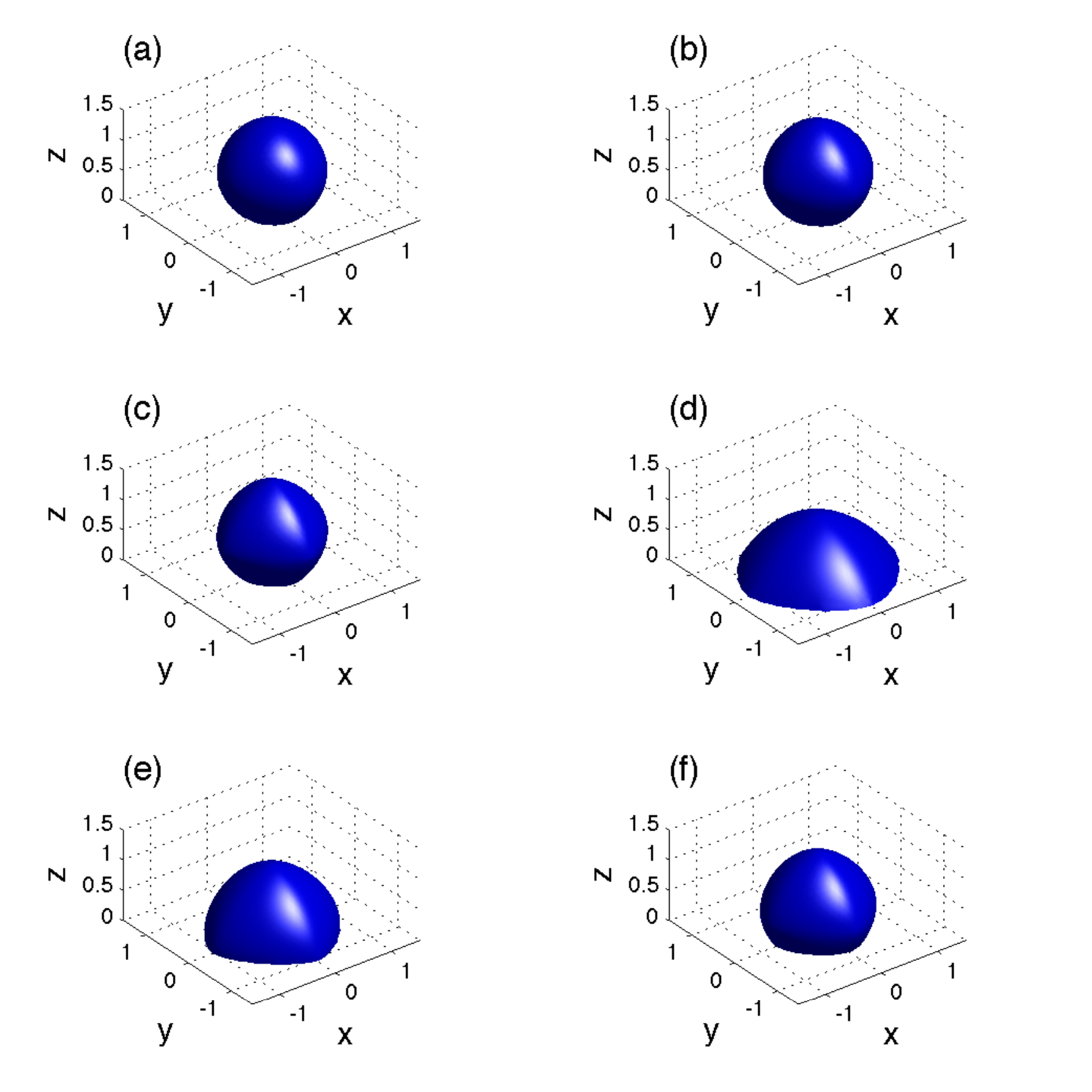}
\caption{The equilibrium geometry of islands under the cubic anisotropy with different material constants $\sigma$. From (a)-(c), the material constant is fixed as $\sigma=\cos(5\pi/6)$, and the degree of anisotropies are chosen as $a=0.1,~0.2,~0.3$; from (d)-(f), the degree of anisotropy is fixed at $a=0.3$, and the material constants are chosen as $\sigma=\cos(\pi/3),~\cos(\pi/2),~\cos(2\pi/3)$. }
\label{fig:Equiasigma}
\end{figure}

By performing numerical simulations, we next examine the equilibrium geometry under different degrees of cubic anisotropy and different material constants. As clearly shown in Fig.~\ref{fig:Equiasigma}(a)-(c), when the degree of the anisotropy is increased from $0.1$ to $0.3$ with a material constant $\sigma=\cos(5\pi/6)$, the equilibrium shape exhibits increasingly sharper and sharper corners. Furthermore, from Fig.~\ref{fig:Equiasigma}(d)-(f), when we change the value of the material constant, we also clearly observe the corresponding change in its equilibrium shape.

Under the cubic surface energy, as expected, the long island film will also exhibit Rayleigh-like instability and could pinch off into small pieces of islands. We consider the evolution of an initial $(1,12,1)$ cuboid island, and the degree of the cubic surface energy is chosen as $a=0.25$, and the material constant is chosen as $\sigma=\cos(2\pi/3)$. As can be seen in Fig.~\ref{fig:weak112islands}, the long cuboid island pinches off, and finally dewets to three pieces of small islands.
\begin{figure}[!htp]
\centering
\includegraphics[width=0.75\textwidth]{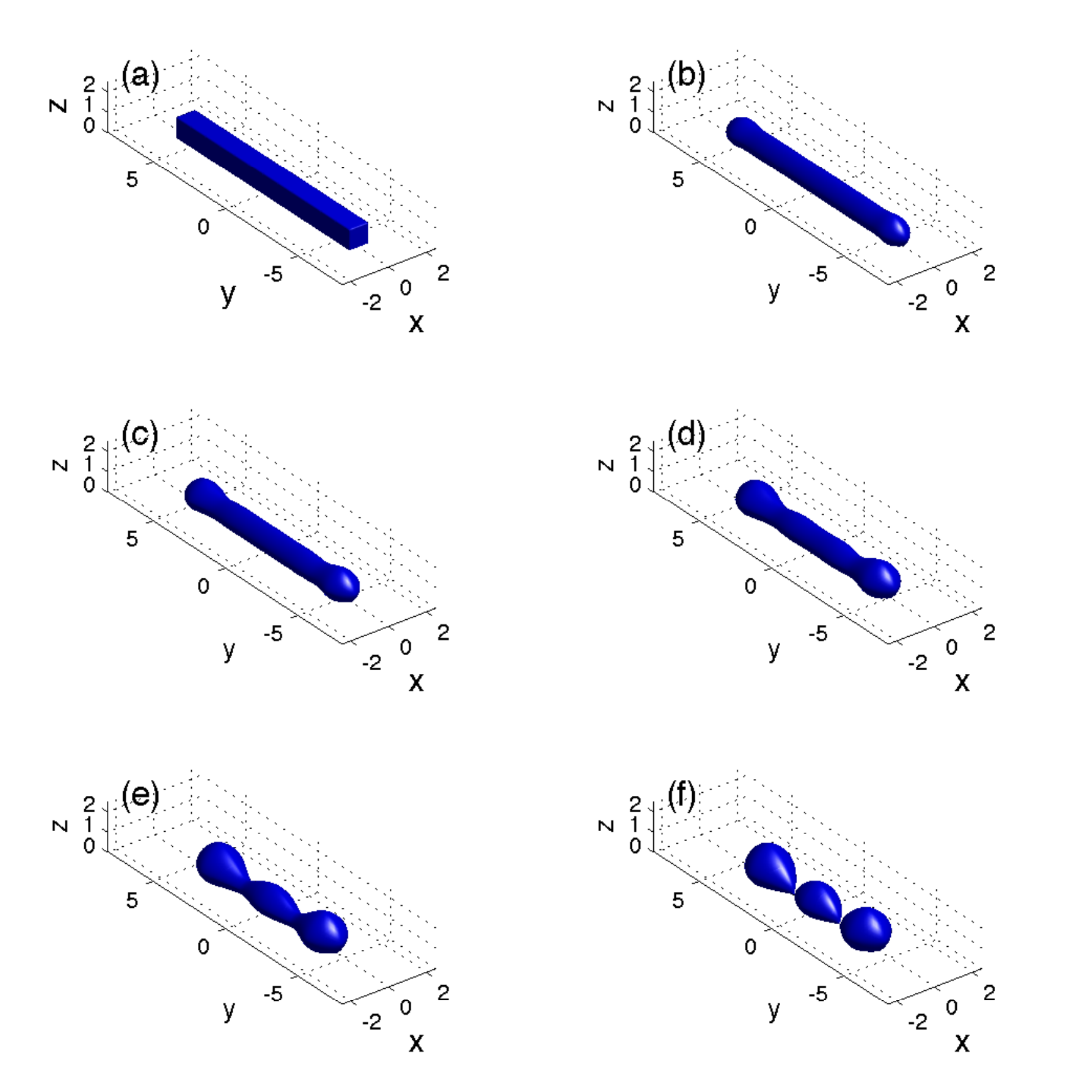}
\caption{Several snapshots in the evolution of an initially $(1,12,1)$ cuboid island until its pinch-off under the cubic anisotropy with $a=0.25$: (a) $t=0$; (b) $t=0.30$; (c) $t=0.60$; (d) $t=0.90$; (e) $t=1.40$; (f) $t=1.58$. The material constant is chosen as $\sigma=\cos(2\pi/3)$.}
\label{fig:weak112islands}
\end{figure}

In addition to the cubic anisotropy, we also perform numerical simulations for the ellipsoidal anisotropy, which is defined as
\begin{equation}
\gamma(\vec n) = \sqrt{a_1^2n_1^2+a_2^2n_2^2+a_3^2n_3^2},
\end{equation}
where $a_1,a_2,a_3$ are the ratio in each direction component. The corresponding equilibrium shape for this type of anisotropy is self-similar to an ellipsoid with semi-major axes $a_1, a_2, a_3$ (see the reference~\cite{Bao18a}), i.e.,
\begin{equation}
\frac{x^2}{a_1^2}+\frac{y^2}{a_2^2}+\frac{z^2}{a_3^2}=1.
\end{equation}
Fig.~\ref{fig:ellipoid} depicts the morphology evolution of an initial cuboid island towards its equilibrium shape. The surface energy anisotropy is chosen as $\gamma(\vec n) = \sqrt{2n_1^2+n_2^2+n_3^2}$.  From the figure, we can see that the island film eventually reaches at its equilibrium with an ellipsoidal shape. This is consistent with the theoretical prediction since the corresponding equilibrium shape for the anisotropy $\gamma(\vec n) = \sqrt{2n_1^2+n_2^2+n_3^2}$ is self-similar to an ellipsoid $\frac{x^2}{2} + y^2 + z^2=1$.

 \begin{figure}[!htp]
\centering
\includegraphics[width=0.75\textwidth]{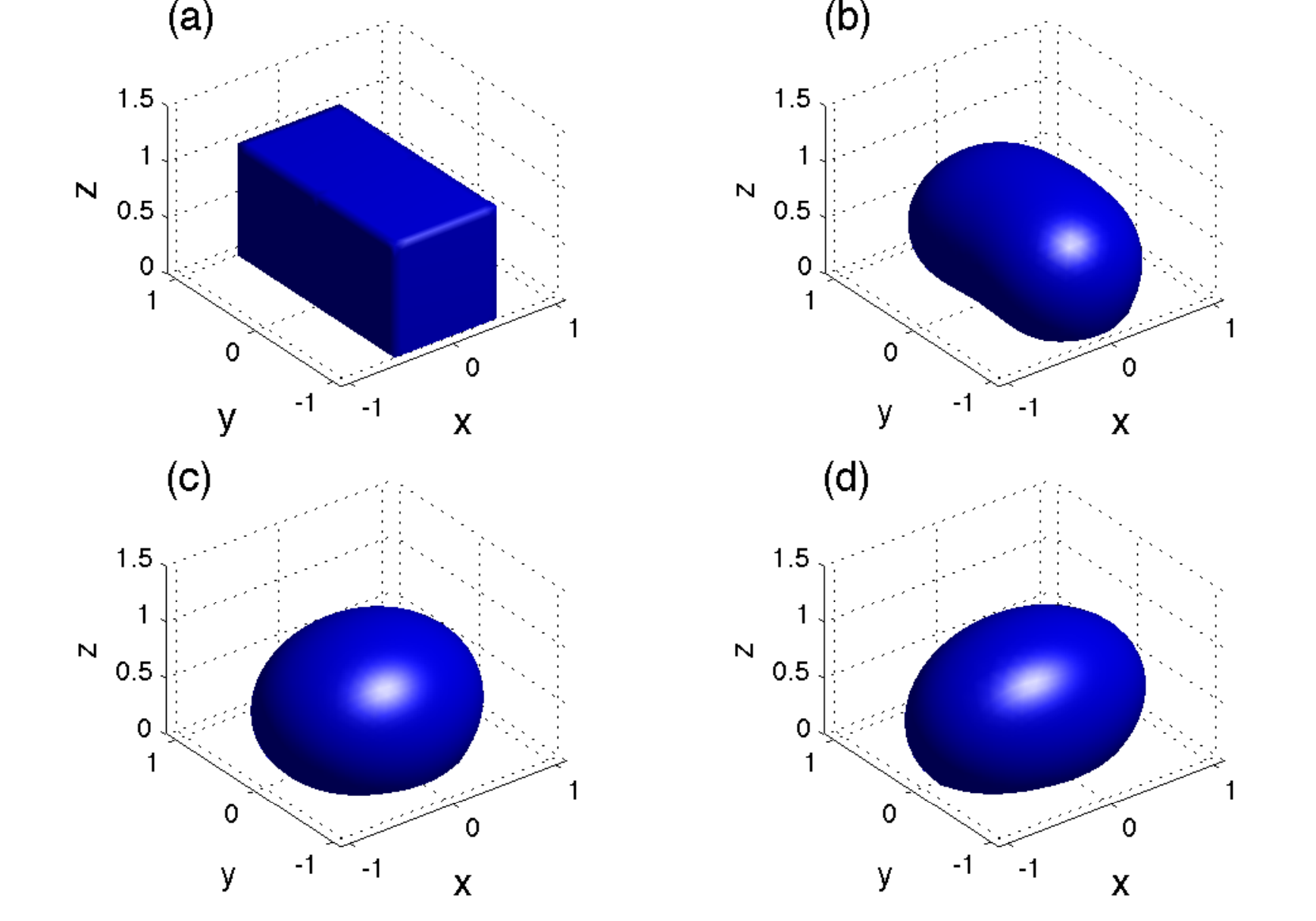}
\caption{Several snapshots in the evolution of an initially $(1,2,1)$ cuboid island towards its equilibrium under the ellipsoidal anisotropy with $a_1=\sqrt{2},a_2=1,a_3=1$: (a) $t=0$; (b) $t=0.01$; (c) $t=0.05$; (d) $t=0.20$, where the material constant is chosen as $\sigma=\cos(3\pi/4)$, and the initial surface mesh consists of 2048 triangles and 1049 vertices with 48 vertices on the boundary, and the time step is uniformly chosen as $\tau_m=10^{-4}$.}
\label{fig:ellipoid}
\end{figure}

\section{Conclusions}

Based on a novel variational formulation in terms of $\boldsymbol{\xi}$-vector formulation, we developed a parametric finite element method for solving solid-state dewetting problems in three dimensions (3D). In each time step, the contact line $\Gamma^{m+1}$ is first updated according to the relaxed contact angle condition; then, by prescribing the boundary curve $\Gamma^{m+1}$ as the explicit boundary condition, the variational formulation is discretized by a semi-implicit parametric finite element method
in order to obtain the new surface $S^{m+1}$. The resulted system is a system of linear and sparse algebra equations, which can be efficiently solved by many existing fast algorithms. We performed ample numerical examples for investigating solid-state dewetting of thin films with isotropic/weakly anisotropic surface energies. We observed that small islands tend to form spherical shapes as the equilibrium in the isotropic case, while long islands could break up into pieces of small isolated islands, and islands with some special geometries exhibit interesting phenomena and complexities.
Numerical results have demonstrated high efficiency and accuracy of the proposed PFEM scheme for solving solid-state dewetting problems with isotropic/weakly anisotropic surface energies in 3D.


\begin{thebibliography}{10}

\bibitem{Amram12}
{\sc D.~Amram, L.~Klinger, and E.~Rabkin}, {\em Anisotropic hole growth during
  solid-state dewetting of single-crystal {Au--Fe} thin films}, Acta Mater., 60
  (2012), pp.~3047--3056.

\bibitem{Bansch04}
{\sc E.~B{\"a}nsch, P.~Morin, and R.~H. Nochetto}, {\em Surface diffusion of
  graphs: variational formulation, error analysis, and simulation}, SIAM J.
  Numer. Anal., 42 (2004), pp.~773--799.

\bibitem{Bansch05}
{\sc E.~B{\"a}nsch, P.~Morin, and R.~H. Nochetto}, {\em A finite element method
  for surface diffusion: the parametric case}, J. Comput. Phys., 203 (2005),
  pp.~321--343.

\bibitem{Bao17b}
{\sc W.~Bao, W.~Jiang, D.~J. Srolovitz, and Y.~Wang}, {\em Stable equilibria of
  anisotropic particles on substrates: a generalized winterbottom
  construction}, SIAM J. Appl. Math, 77 (2017), pp.~2093--2118.

\bibitem{Bao17}
{\sc W.~Bao, W.~Jiang, Y.~Wang, and Q.~Zhao}, {\em A parametric finite element
  method for solid-state dewetting problems with anisotropic surface energies},
  J. Comput. Phys., 330 (2017), pp.~380--400.

\bibitem{Barrett07Ani}
{\sc J.~W. Barrett, H.~Garcke, and R.~N{\"u}rnberg}, {\em Numerical
  approximation of anisotropic geometric evolution equations in the plane},
  SIMA J. Numer. Anal, 28 (2007), pp.~292--330.

\bibitem{Barrett07b}
{\sc J.~W. Barrett, H.~Garcke, and R.~N{\"u}rnberg}, {\em On the variational
  approximation of combined second and fourth order geometric evolution
  equations}, SIAM J. Sci. Comput, 29 (2007), pp.~1006--1041.

\bibitem{Barrett07}
{\sc J.~W. Barrett, H.~Garcke, and R.~N{\"u}rnberg}, {\em A parametric finite
  element method for fourth order geometric evolution equations}, J. Comput.
  Phys, 222 (2007), pp.~441--467.

\bibitem{Barrett08JCP}
{\sc J.~W. Barrett, H.~Garcke, and R.~N{\"u}rnberg}, {\em On the parametric
  finite element approximation of evolving hypersurfaces in $\mathbb{R}^3$}, J.
  Comput. Phys, 227 (2008), pp.~4281--4307.

\bibitem{Barrett08}
{\sc J.~W. Barrett, H.~Garcke, and R.~N{\"u}rnberg}, {\em Parametric
  approximation of {Willmore} flow and related geometric evolution equations},
  SIAM J. Sci. Comput, 31 (2008), pp.~225--253.

\bibitem{Barrett08Ani}
{\sc J.~W. Barrett, H.~Garcke, and R.~N{\"u}rnberg}, {\em A variational
  formulation of anisotropic geometric evolution equations in higher
  dimensions}, Numer. Math, 109 (2008), pp.~1--44.

\bibitem{Barrett10}
{\sc J.~W. Barrett, H.~Garcke, and R.~N{\"u}rnberg}, {\em Finite-element
  approximation of coupled surface and grain boundary motion with applications
  to thermal grooving and sintering}, Eur. J. Appl. Math, 21 (2010),
  pp.~519--556.

\bibitem{Cahn74}
{\sc J.~Cahn and D.~Hoffman}, {\em A vector thermodynamics for anisotropic
  surfaces: I. curved and faceted surfaces}, Acta Metall., 22 (1974),
  pp.~1205--1214.

\bibitem{Cahn94}
{\sc J.~W. Cahn and J.~E. Taylor}, {\em Surface motion by surface diffusion},
  Acta Metall. Mater., 42 (1994), pp.~1045--1063.

\bibitem{deGennes85}
{\sc P.-G. De~Gennes}, {\em Wetting: statics and dynamics}, Rev. Mod. Phys., 57
  (1985), pp.~827--863.

\bibitem{Deckelnick05}
{\sc K.~Deckelnick, G.~Dziuk, and C.~M. Elliott}, {\em Computation of geometric
  partial differential equations and mean curvature flow}, Acta Numer., 14
  (2005), pp.~139--232.

\bibitem{Dornel06}
{\sc E.~Dornel, J.~Barbe, F.~De~Cr{\'e}cy, G.~Lacolle, and J.~Eymery}, {\em
  Surface diffusion dewetting of thin solid films: Numerical method and
  application to {Si/SiO$_2$}}, Phys. Rev. B, 73 (2006), p.~115427.

\bibitem{Du10}
{\sc P.~Du, M.~Khenner, and H.~Wong}, {\em A tangent-plane marker-particle
  method for the computation of three-dimensional solid surfaces evolving by
  surface diffusion on a substrate}, J. Comput. Phys., 229 (2010),
  pp.~813--827.

\bibitem{Dziuk90}
{\sc G.~Dziuk}, {\em An algorithm for evolutionary surfaces}, Numer. Math, 58
  (1990), pp.~603--611.

\bibitem{Dziuk13}
{\sc G.~Dziuk and C.~M. Elliott}, {\em Finite element methods for surface
  {PDEs}}, Acta Numer., 22 (2013), pp.~289--396.

\bibitem{Hausser05}
{\sc F.~Hausser and A.~Voigt}, {\em A discrete scheme for regularized
  anisotropic surface diffusion: a 6th order geometric evolution equation},
  Interfaces Free Bound, 7 (2005), pp.~353--370.

\bibitem{Hausser07}
{\sc F.~Hausser and A.~Voigt}, {\em A discrete scheme for parametric
  anisotropic surface diffusion}, J. Sci. Comput, 30 (2007), pp.~223--235.

\bibitem{Hildebrandt12}
{\sc S.~Hildebrandt and H.~Karcher}, {\em Geometric analysis and nonlinear
  partial differential equations}, Springer Science \& Business Media, 2012.

\bibitem{Hoffman72}
{\sc D.~W. Hoffman and J.~W. Cahn}, {\em A vector thermodynamics for
  anisotropic surfaces: I. fundamentals and application to plane surface
  junctions}, Surface Science, 31 (1972), pp.~368--388.

\bibitem{Hon14}
{\sc S.~Y. Hon, S.~Leung, and H.~Zhao}, {\em A cell based particle method for
  modeling dynamic interfaces}, J. Comput. Phys., 272 (2014), pp.~279--306.

\bibitem{Huang19b}
{\sc Q.-A. Huang, W.~Jiang, and J.~Z. Yang}, {\em An efficient and
  unconditionally energy stable scheme for simulating solid-state dewetting of
  thin films with isotropic surface energy}, accepted by Commu. Comput. Phys.,
  (2019).

\bibitem{Jiang12}
{\sc W.~Jiang, W.~Bao, C.~V. Thompson, and D.~J. Srolovitz}, {\em Phase field
  approach for simulating solid-state dewetting problems}, Acta Mater., 60
  (2012), pp.~5578--5592.

\bibitem{Jiang16}
{\sc W.~Jiang, Y.~Wang, Q.~Zhao, D.~J. Srolovitz, and W.~Bao}, {\em Solid-state
  dewetting and island morphologies in strongly anisotropic materials}, Scripta
  Mater., 115 (2016), pp.~123--127.

\bibitem{Jiang18}
{\sc W.~Jiang and Q.~Zhao}, {\em Sharp-interface approach for simulating
  solid-state dewetting in two dimensions: a {Cahn-Hoffman}
  $\boldsymbol{\xi}$-vector formulation}, Physica D: Nonlinear Phenomena, 390
  (2019), pp.~69--83.

\bibitem{Bao18a}
{\sc W.~Jiang, Q.~Zhao, and W.~Bao}, {\em Sharp-interface approach for
  simulating solid-state dewetting in three dimensions}, arXiv:1902.05272,
  (2019).

\bibitem{Jiang19}
{\sc W.~Jiang, Q.~Zhao, T.~Qian, D.~J. Srolovitz, and W.~Bao}, {\em Application
  of {Onsager's} variational principle to the dynamics of a solid toroidal
  island on a substrate}, Acta Mater., 163 (2019), pp.~154--160.

\bibitem{Kim15}
{\sc G.~H. Kim and C.~V. Thompson}, {\em Effect of surface energy anisotropy on
  {Rayleigh-like} solid-state dewetting and nanowire stability}, Acta Mater.,
  84 (2015), pp.~190--201.

\bibitem{Kim13}
{\sc G.~H. Kim, R.~V. Zucker, J.~Ye, W.~C. Carter, and C.~V. Thompson}, {\em
  Quantitative analysis of anisotropic edge retraction by solid-state dewetting
  of thin single crystal films}, J. Appl. Phys., 113 (2013), p.~043512.

\bibitem{Kovalenko17}
{\sc O.~Kovalenko, S.~Szab{\'o}, L.~Klinger, and E.~Rabkin}, {\em Solid state
  dewetting of polycrystalline {Mo} film on sapphire}, Acta Mater., 139 (2017),
  pp.~51--61.

\bibitem{Leroy16}
{\sc F.~Leroy, F.~Cheynis, Y.~Almadori, S.~Curiotto, M.~Trautmann,
  J.~Barb{\'e}, P.~M{\"u}ller, et~al.}, {\em How to control solid state
  dewetting: A short review}, Surface Science Reports, 71 (2016), pp.~391--409.

\bibitem{Leung11}
{\sc S.~Leung, J.~Lowengrub, and H.~Zhao}, {\em A grid based particle method
  for solving partial differential equations on evolving surfaces and modeling
  high order geometrical motion}, J. Comput. Phys, 230 (2011), pp.~2540--2561.

\bibitem{Mizsei93}
{\sc J.~Mizsei}, {\em Activating technology of {SnO$_2$} layers by metal
  particles from ultrathin metal films}, Sensors and Actuators B: Chemical, 16
  (1993), pp.~328--333.

\bibitem{Mullins57}
{\sc W.~W. Mullins}, {\em Theory of thermal grooving}, J. Appl. Phys., 28
  (1957), pp.~333--339.

\bibitem{Naffouti16}
{\sc M.~Naffouti, T.~David, A.~Benkouider, L.~Favre, A.~Delobbe, A.~Ronda,
  I.~Berbezier, and M.~Abbarchi}, {\em Templated solid-state dewetting of thin
  silicon films}, Small, 12 (2016), pp.~6115--6123.

\bibitem{Pozzi08}
{\sc P.~Pozzi}, {\em Anisotropic mean curvature flow for two-dimensional
  surfaces in higher codimension: a numerical scheme}, Interface Free Bound.,
  10 (2008), pp.~539--576.

\bibitem{Qian06}
{\sc T.~Qian, X.-P. Wang, and P.~Sheng}, {\em A variational approach to moving
  contact line hydrodynamics}, J. Fluids Mech., 564 (2006), pp.~333--360.

\bibitem{Rabkin14}
{\sc E.~Rabkin, D.~Amram, and E.~Alster}, {\em Solid state dewetting and stress
  relaxation in a thin single crystalline {Ni} film on sapphire}, Acta Mater.,
  74 (2014), pp.~30--38.

\bibitem{Randolph07}
{\sc S.~Randolph, J.~Fowlkes, A.~Melechko, K.~Klein, H.~Meyer~III, M.~Simpson,
  and P.~Rack}, {\em Controlling thin film structure for the dewetting of
  catalyst nanoparticle arrays for subsequent carbon nanofiber growth},
  Nanotechnology, 18 (2007), p.~465304.

\bibitem{Rayleigh78}
{\sc L.~Rayleigh}, {\em On the instability of jets}, Proc.~Lond. Math. Soc, 1
  (1878), pp.~4--13.

\bibitem{Schmidt09}
{\sc V.~Schmidt, J.~V. Wittemann, S.~Senz, and U.~G{\"o}sele}, {\em Silicon
  nanowires: a review on aspects of their growth and their electrical
  properties}, Adv. Mater, 21 (2009), pp.~2681--2702.

\bibitem{Srolovitz86a}
{\sc D.~J. Srolovitz and S.~A. Safran}, {\em Capillary instabilities in thin
  films: {I.} {Energetics}}, J. Appl. Phys., 60 (1986), pp.~247--254.

\bibitem{Taylor92}
{\sc J.~E. Taylor}, {\em {II} -- mean curvature and weighted mean curvature},
  Acta Metall. Mater, 40 (1992), pp.~1475--1485.

\bibitem{Thompson12}
{\sc C.~V. Thompson}, {\em Solid-state dewetting of thin films}, Annu. Rev.
  Mater. Res., 42 (2012), pp.~399--434.

\bibitem{Wang15}
{\sc Y.~Wang, W.~Jiang, W.~Bao, and D.~J. Srolovitz}, {\em Sharp interface
  model for solid-state dewetting problems with weakly anisotropic surface
  energies}, Phys. Rev. B, 91 (2015), p.~045303.

\bibitem{Winterbottom67}
{\sc W.~Winterbottom}, {\em Equilibrium shape of a small particle in contact
  with a foreign substrate}, Acta Metall., 15 (1967), pp.~303--310.

\bibitem{Wong00}
{\sc H.~Wong, P.~Voorhees, M.~Miksis, and S.~Davis}, {\em Periodic mass
  shedding of a retracting solid film step}, Acta Mater., 48 (2000),
  pp.~1719--1728.

\bibitem{Xu10}
{\sc X.~Xu and X.~Wang}, {\em Derivation of the {Wenzel} and {Cassie} equations
  from a phase field model for two phase flow on rough surface}, SIAM J. Appl.
  Math., 70 (2010), pp.~2929--2941.

\bibitem{Xu11}
{\sc X.~Xu and X.~Wang}, {\em Analysis of wetting and contact angle hysteresis
  on chemically patterned surfaces}, SIAM J. Appl. Math., 71 (2011),
  pp.~1753--1779.

\bibitem{Xu09}
{\sc Y.~Xu and C.-W. Shu}, {\em Local discontinuous {Galerkin} method for
  surface diffusion and {Willmore} flow of graphs}, J. Sci. Comput, 40 (2009),
  pp.~375--390.

\bibitem{Ye10b}
{\sc J.~Ye and C.~V. Thompson}, {\em Regular pattern formation through the
  retraction and pinch-off of edges during solid-state dewetting of patterned
  single crystal films}, Phys. Rev. B, 82 (2010), p.~193408.

\bibitem{Ye11a}
{\sc J.~Ye and C.~V. Thompson}, {\em Anisotropic edge retraction and hole
  growth during solid-state dewetting of single crystal nickel thin films},
  Acta Mater., 59 (2011), pp.~582--589.

\bibitem{Ye11b}
{\sc J.~Ye and C.~V. Thompson}, {\em Templated solid-state dewetting to
  controllably produce complex patterns}, Adv. Mater., 23 (2011),
  pp.~1567--1571.

\bibitem{Zhao19}
{\sc Q.~Zhao}, {\em A sharp-interface model and its numerical approximation for
  solid-state dewetting with axisymmetric geometry}, J. Comput. Appl. Math.,
  361 (2019), pp.~144--156.

\bibitem{Zhao17}
{\sc Q.~Zhao, W.~Jiang, D.~J. Srolovitz, and W.~Bao}, {\em Triple junction drag
  effects during topological changes in the evolution of polycrystalline
  microstructures}, Acta Mater., 128 (2017), pp.~345--350.

\bibitem{Zucker16}
{\sc R.~V. Zucker, G.~H. Kim, J.~Ye, W.~C. Carter, and C.~V. Thompson}, {\em
  The mechanism of corner instabilities in single-crystal thin films during
  dewetting}, J. Appl. Phys., 119 (2016), p.~125306.

\end{thebibliography}
\end{document}